\documentclass[11pt]{article} \usepackage{amsmath,amsthm,amsfonts,amssymb,amscd,amsbsy}
\usepackage{graphicx}
\usepackage{array}
\usepackage{enumerate}
\usepackage[nodayofweek]{datetime}
\usepackage[english]{babel}
\usepackage{mathtools}
\usepackage[toc,page]{appendix}
\usepackage{verbatim} 
\usepackage{amsthm} 
\usepackage{enumitem}
\usepackage{enumerate}
\usepackage{bbm} 
\usepackage[normalem]{ulem} 
\usepackage{bm}
\usepackage{authblk}
\usepackage[font=small,labelfont=bf]{caption}

\usepackage[utf8]{inputenc}
\usepackage[pdfpagelabels]{hyperref}

\usepackage{mathrsfs}
\usepackage{amsmath}
\usepackage{amssymb}
\usepackage{amsfonts}
\usepackage{amsthm,bbm}
\usepackage{graphicx}
\usepackage[margin=2.5cm,]{geometry}

\makeatletter
\@addtoreset{equation}{myequation}
\makeatother

\usepackage{hyperref}
\hypersetup{colorlinks=true,allcolors=[rgb]{0,0,0.6}}

\newtheorem{thm}{Theorem}
\newtheorem{lm}[thm]{Lemma}

\theoremstyle{definition}

\newcommand{\be}{\begin}
\newcommand{\e}{\end}
\newcommand{\beq}{\begin{equation}}
\newcommand{\eeq}{\end{equation}}

\numberwithin{equation}{section}
\numberwithin{thm}{section}

\newcommand{\Z}{{\mathbb Z}}

\newcommand{\C}{{\mathbb C}}

\newcommand{\Lam}{{\Lambda}}
\newcommand{\gam}{{\gamma}}

\renewcommand{\l}{\left}
\renewcommand{\r}{\right}

\newcommand{\Hm}[1]{\leavevmode{\marginpar{\tiny%
$\hbox to 0mm{\hspace*{-0.5mm}$\leftarrow$\hss}%
\vcenter{\vrule depth 0.1mm height 0.1mm width \the\marginparwidth}%
\hbox to 0mm{\hss$\rightarrow$\hspace*{-0.5mm}}$\\\relax\raggedright
#1}}}

 \usepackage{amssymb}
 \usepackage{comment}
 \usepackage{graphics}
\usepackage{bbm}

\usepackage{dsfont}
\usepackage{tikz}
\usepackage{tikz-cd}
\usepackage{tkz-euclide}
\usepackage{hyperref}
\usepackage{siunitx} 
\begin{document}

 \title{Quantitatively improved finite-size criteria for spectral gaps} 
 \author[1]{Marius Lemm}
 \author[2]{David Xiang}

  \affil[1]{Department of Mathematics, University of Tübingen,  72076 Tübingen, Germany}
  \affil[2]{Harvard College, University Hall, Cambridge, MA 02138, USA}  
  \date{December 14, 2021}

\renewcommand\Affilfont{\itshape\small}
	 	 \maketitle
 \begin{abstract}
Finite-size criteria have emerged as an effective tool for deriving spectral gaps in higher-dimensional frustration-free quantum spin systems. We quantitatively improve the existing finite-size criteria by introducing a novel subsystem weighting scheme. The approach applies to Euclidean lattices of any dimension, the honeycomb lattice, and the triangular lattice.


\end{abstract}
\maketitle

\section{Introduction}
In quantum lattice systems, the existence of a spectral gap above the ground state sector places strong constraints on the system's low-energy behavior. For instance, a spectral gap controls ground state correlation in various sense \cite{AAG,ALVV,H,NS} and is central to the classification of topological quantum phases \cite{BMNS,H_LSM}. However, the derivation of a spectral gap is often challenging because excitation energies over the entire Hilbert space need to be controlled. In the better-behaved class of frustration-free quantum spin systems, excitation energies are strictly local and so analytical tools for deriving gaps are available. The list of techniques includes the martingale method \cite{N} and finite-size criteria, which come in two main flavors: those studying the angle between subspaces based on a duality lemma of Fannes-Nachtergaele-Werner \cite{FNW}, see also \cite{SS}, and those distilled by Knabe \cite{K} from the highly influential work of Affleck-Kennedy-Lieb-Tasaki on their eponymous spin chain \cite{AKLT}. 

Starting with the work of Gosset-Mozgunov \cite{GM}, finite-size criteria have seen considerable progress in the last years, especially in the context of higher-dimensional quantum spin systems.
 This progress concerns both the methodological side \cite{Anshu,L2,LM} and the applications side \cite{Aetal,JL,L1,LN,LSW,NWY,PW1,PW2,WY}. Recent methodological advances have mostly focused on the scaling of the gap threshold (called $t_\ell$ in Eq.\ \eqref{eq:generic} below) in $\ell$, which is of essential theoretical importance. Less attention has been placed on producing thresholds with good constants. This can be a limiting factor when attempting the verification of a finite-size criterion in practice by an explicit calculation of the gap of a particular subsystem. \textit{The main goal of the present paper is to remedy this situation by substantially improving the constants entering the threshold $t_\ell$.}  
 
Our first main result concerns the gap thresholds on the Euclidean lattice in any dimension and is formally stated in Theorem \ref{thm:main} below. Our main contribution is a new weighting scheme which has the benefit of adapting seamlessly to dimensions $\geq 3$ where deriving the optimal threshold scaling $\ell^{-2}$ has so far been inaccessible to the weighting approach of \cite{GM} and instead required sophisticated tools form quantum information theory \cite{Anshu}. The heuristic idea underlying our weighting scheme is that tensor-like product structures allow for computationally efficient and dimensionally robust encodings of the weighting scheme, while at the same time allowing to effectively implement the principle that the weight of an interaction term should increase with its distance to the boundary of the subsystem. 

We also provide finite-size criteria for the honeycomb and triangular lattices in Theorems \ref{thm:mainhoney} and \ref{thm:maintri} respectively. For both of these lattices, we obtain for the first time the optimal inverse-square threshold scaling $\ell^{-2}$ which should be sharp in view of the excitation energy of spin wave trial states for the ferromagnetic Heisenberg Hamiltonians on these respective graphs.

\section{Setup and main results}
Consider a quantum spin system in which one qudit $\C^d$ is placed at each site of the box
$$
\Lam_L=([0,L]\cap\Z)^d
$$
so the total Hilbert space of the system is $\bigotimes_{x \in \Lambda_L} \mathbb{C}^d$. On it, we consider a Hamiltonian with nearest-neighbor interactions,
$$
H_L^{\mathrm{per}}=\sum_{x \in \Lambda_L}\sum_{1 \leq j \leq D} h_{x,e_j}
$$
and with periodic boundary conditions. The local interaction is given by
$$
h_{x,e_j} = h \otimes \mathrm{Id}_{\Lambda_L \backslash\{x, x+e_j\}}
$$
for a fixed projector $h:\C^d\otimes \C^d \to 	\C^d\otimes \C^d$. Note that the edges are directed and the local interaction $h_{x,e_j}$ is allowed to depend on the orientation of the edge. 
The fact that $h$ is identical across $\Lam_L$ makes $H_L$ a translation-invariant Hamiltonian. 
\be{ass}\label{ass:FF}
$H_{\Lam_L}$ is frustration-free for every $L\geq 1$, i.e., $\inf\mathrm{spec}\, H_{\Lam_L}=0$.
\e{ass}
We recall that there exist many physically relevant example of frustration-free Hamiltonians, e.g., ferromagnetic Heisenberg models, AKLT models, and Motzkin spin chains.
We denote by $\gamma^{\mathrm{per}}_L$ the spectral gap of $H_{}$, which by frustration-freeness equals the smallest positive eigenvalue of $H_{\Lam_L}$, i.e.,
$$
\gamma_L^{\mathrm{per}}=\inf \left(\mathrm{spec}H_{\Lam_L}\setminus \{0\}\right).
$$

For $\ell\geq 1$, we similarly define the Hamiltonian
$$
H_{\Lam_\ell} = \sum_{x \in \Lambda_\ell}\sum_{1 \leq j \leq D} h_{x,e_j}
$$
with open boundary conditions and we denote its spectral gap by $\gamma_\ell$. Note that Assumption \ref{ass:FF} ensures that $H_{\Lam_\ell}$ is frustration-free as well.

\subsection{Finite-size criteria on Euclidean lattices}
We briefly recall the existing finite-size criteria on Euclidean (or hypercubic) lattices of spatial dimension $D$ and state our main result. We consider finite-size criteria of the generic form
\beq\label{eq:generic}
\gam_L^{\mathrm{per}}\geq c_\ell \l(\gam_\ell -t_\ell\r)
\eeq
with $L$ sufficiently large compared to $\ell$. In practice, this can be applied as follows: If there exists a small system size $\ell$ such that by an explicit calculation one can show $\gam_\ell>t_\ell$, then one obtains $\gam_L^{\mathrm{per}}\geq c>0$ with $c$ independent of $L$, i.e., a gap in the infinite-volume limit. We see that the constant $c_\ell$ is largely irrelevant for the intended use of the criterion. By contrast, the numerical value of the so-called \textit{local gap threshold} $t_\ell$ is absolutely critical.

Finite-size criteria of the form \eqref{eq:generic} with the following thresholds $t_\ell$ have been derived. 

\be{itemize}
\item Knabe \cite{K} proved \eqref{eq:generic} in one dimension with threshold $t_\ell=\frac{1}{\ell}$.
\item Gosset-Mozgunov \cite{GM} proved \eqref{eq:generic} in one dimension with $t_\ell=\frac{6}{(\ell+1)(\ell+2)}\sim \frac{6}{\ell^2}$ for $\ell$ large and in two dimensions with threshold $t_\ell=\frac{8}{\ell^2}$ for slightly modified regions. See \cite{LM} for a related result for open boundary conditions.
\item \cite{L2} proved \eqref{eq:generic} in any dimension with threshold $t_\ell=\frac{3}{\ell}$. 
\item Anshu \cite{Anshu} proved \eqref{eq:generic} in any dimension with threshold $t_\ell=\frac{D^2 6^D 800 (4D-2)^2}{\ell^2}$.
\e{itemize}

Here, ``dimension'' always refers to the spatial dimension $D$. The local Hilbert space dimension $d$ can be arbitrary. Regarding conventions, we note that in the present paper, we use the convention that $\ell$ counts the number of bonds along each side of $\Lam_\ell$; this is $1$ less than the number of sites which is used in some of the works mentioned above.

There exist bounds on the best possible threshold asymptotic as $\ell\to\infty$. Namely, the spin-$1/2$ Heisenberg ferromagnet on the Euclidean lattice has spectral gap upper bounded by $1-\cos(\pi/\ell)\sim \frac{\pi^2}{2\ell^2}$ by spin wave theory. This shows that the threshold scaling necessarily satisfies $t_\ell\gtrsim \frac{\pi^2}{2\ell^2}\approx \frac{4.93}{\ell^2}$ as $\ell\to\infty$ in any dimension. In particular, the inverse-square behavior of $t_\ell$ that has been found is optimal.

Our first main result reads as follows.

\be{thm}[Main result on Euclidean lattices]\label{thm:main}
Assume that $H_{\Lam_L}$ is frustration-free. Let $D\geq 2$ and let $L \geq 5$, $2\leq \ell<L/2$. Then,
\begin{equation} \label{eqn:thm1}
    \gamma_L^{\mathrm{per}} \geq \left(\frac{5}{6}\right)^D \left( \gamma_\ell- t_\ell \right)
\end{equation}
where 
\beq\label{eq:tlmain}
t_\ell=\left(\frac{6}{5}\right)^D\left(\frac{5}{\ell^2}+\frac{300}{\ell^3}\r),\qquad \ell\geq 10.
\eeq
For $2\leq \ell\leq 9$, the values of $t_\ell$ are displayed  in Table \ref{tab:table1} for $D=2$ and in Table \ref{tab:table2} for $D=3$.
\e{thm}

\begin{table}[h]
  \begin{center}
    \begin{tabular}{c|cccccccc} 
        $\ell$       &2     & 3 &   4        &5     &6      &7      &8        &9\\ \hline
                $t_\ell$& 0.667 & 0.395 & 0.257 & 0.181 & 0.134 & 0.104 &0.082 & 0.067 \\ 
        $\ell^2 t_\ell$& 2.667 & 3.547 & 4.102 &4.503 & 4.808    &5.049  &5.245  &  5.407
    \end{tabular}
        \caption{Values of $t_\ell$ for $D=2$ and $2\leq \ell\leq 9$.}
            \label{tab:table1}
  \end{center}
\end{table}

\begin{table}[h]
  \begin{center}
    \begin{tabular}{c|cccccccc}  
        $\ell$ & 2& 3 & 4&5&6&7&8&9\\ \hline
        $t_\ell$ & .667 & 0.400 & 0.264 & 0.188 & 0.141 & 0.110 &0.088 & 0.073 \\ \hline
        $\ell^2t_\ell$ & 2.667 & 3.593 & 4.212&4.685  &5.061  &5.369   &5.627 &5.846
    \end{tabular}
     \caption{Values of $t_\ell$ for $D=3$ and $2\leq \ell\leq 9$.}
         \label{tab:table2}
  \end{center}
\end{table}

\be{rmk}\label{rmk:main}
\be{enumerate}[label=(\roman*)]
\item

In formula \eqref{eqn:thm1} for the threshold $t_\ell$, we took care to obtain the best leading asymptotic afforded by the proof method, namely $\left(\frac{6}{5}\right)^D 5\ell^{-2}$. By contrast, the error term is influenced by various choices made in the proof about how to precisely implement error estimates and accordingly depends directly on the minimal $\ell$ for which the result applies. The above choice of having a formula for $t_\ell$ valid for all $\ell\geq 10$ with constant $300$ in the error term is a compromise between these competing effects. 
\item 

For fixed small $\ell$, the thresholds can be significantly improved by using our procedure but optimizing over parameters well-adapted to the pertinent value of $\ell$. For the potential benefit of practitioners interested in applying our refined finite-size criterion to concrete systems, these significantly better thresholds for $2\leq \ell\leq 9$ are recorded separately in Tables \ref{tab:table1} and \ref{tab:table2}. Analogs of Tables \ref{tab:table1} and \ref{tab:table2} in dimensions $D\geq 4$ can also be computed by implementing the procedure laid out in the proof of Theorem \ref{thm:main}.
\item 
 In dimensions $D\geq 3$, the quantitative improvement afforded by Theorem \ref{thm:main} over existing thresholds that were reviewed above is substantial.
Also in 2D, we obtain a slight asymptotic improvement over the best previously known threshold which satisfies $\ell^2 t_\ell=8$ \cite{GM}. We obtain $\ell^2 t_\ell \sim \frac{36}{5}=7.2$ as $\ell\to\infty$. As the example of the Heisenberg ferromagnet shows, one cannot do better than $\ell^2 t_\ell\geq \frac{\pi^2}{2}\approx 4.93$ for $\ell$ large, so the room for further asymptotic improvements in 2D is limited.

%
%
\e{enumerate}
\e{rmk}

\subsection{Discussion}
To place this result in context, we first recall that the work \cite{GM} introduced a refinement of Knabe's approach via subsystem Hamiltonians \cite{K} by suitably weighting the local interactions. The weighting scheme used in \cite{GM} works well in the two-dimensional setting considered there, but has so far resisted extension to dimensions $>2$ for technical reasons. By contrast, \cite{Anshu} proceeds completely differently in higher dimensions via an ingenious use of gap amplification and the detectability lemma from quantum information theory. Until now, \cite{Anshu} was the only proof of the optimal threshold scaling $t_\ell\sim \ell^{-2}$ in higher dimensions. A disadvantage of the detectability lemma approach is that it comes at the price of relatively large constants which can create difficulties in concrete applications where quantitative values matter. It has therefore been of practical interest to find a way to adapt the weighting method to higher dimensions.

 The present work shows that the weighting method can be modified to apply in higher dimensions. Our main contribution is to develop a novel weighting scheme which extends to higher dimensions. The main idea is to define the weights through a certain tensor-like product structure; see \eqref{eq:WBgrid} as well as Figures \ref{fig:WB3grid} and \ref{fig:WB3gridproduct}. Implementing this scheme yields the first proof of the optimal threshold scaling in higher dimensions that does not require the detectability lemma. The weighting scheme efficiently implements the general principle that boundary interactions receive lower weights. As a result, our finite-size criterion has the smallest thresholds $t_\ell$ that are available in the literature, which makes it easier to verify in concrete applications.

We close the discussion by listing possible extensions and an open problem.

\be{rmk}
\be{enumerate}[label=(\roman*)]
\item
The assumption that $h$ is a projector is mild. If $h$ is not a projector, we may scale it appropriately and then replace it with a projector onto the orthogonal complement of $\ker h$. The spectral gap of the new Hamiltonian differs only by a constant multiplicative factor from the old Hamiltonian. 
\item 
The proof also applies if the assumptions of rotation- and/or translation-invariance on the interaction are removed. The result is a modified finite-size criterion in which different rotations and/or translations of subsystems enter on the right-hand side.
\item It is an open problem, though one of purely theoretical nature, whether the threshold scaling $t_\ell\sim \ell^{-2}$ can be derived with dimension-independent constant as was possible for $\ell^{-1}$-scaling \cite{L2}.
\e{enumerate}
\e{rmk}

\subsection{Finite-size criterion for the honeycomb lattice}
Motivated by the AKLT conjecture that the honeycomb AKLT model exhibits a spectral gap \cite{AKLT}, Knabe investigated a finite-size criterion on the honeycomb lattice already in 1988 \cite{K}, albeit inconclusively. Recent years have seen a number of applications of finite-size criteria to the honeycomb and decorated honeycomb lattices \cite{Aetal,LSW,PW1,PW2}, in some cases combine with extensive numerical calculations. Here we derive for the first time a general finite-size criterion of the form \eqref{eq:generic} for the honeycomb lattice which has the expected inverse-square threshold scaling just like the Euclidean case studied in Theorem \ref{thm:main} but naturally with different constants.

To formulate the result, we let $\mathbb H_L$ denote the hexagonal lattice wrapped on an $L\times L$ torus. We write 
$$
H^{\mathrm{per}}_{\mathbb H_L}=\sum_{\substack{x,y\in \mathbb H_L\\ x\sim y}} h_{x,y}
$$
for the associated Hamiltonian. Since there is no natural orientation of the edges, we assume that the local interaction commutes with the swap operator defined by $S(|\psi\rangle\otimes |\phi\rangle)=(|\phi\rangle\otimes |\psi\rangle)$. We assume that $H^{\mathrm{per}}_{\mathbb H_L}$ is frustration-free and write $\gamma_L^{\mathrm{per}}$ for its spectral gap. 

The relevant finite subsystems are $\ell \times \ell$ ``slanted" grids of hexagons called $\mathcal B_{\ell}$ as depicted in Figure \ref{fig:hexslant} below. The associated Hamiltonians are denoted by $H_{\mathcal{B}_\ell}$ with spectral gaps $\gamma_\ell$. 

\vspace{.5cm}
\begin{figure}[h]
\begin{center}
\resizebox{4cm}{!}{
\begin{tikzpicture}
\node at (1.2990,-0.7500) [circle,fill,inner sep=1.5pt]{};
\node at (1.2990,0.7500) [circle,fill,inner sep=1.5pt]{};
\node at (0.0000,1.5000) [circle,fill,inner sep=1.5pt]{};
\node at (-1.2990,0.7500) [circle,fill,inner sep=1.5pt]{};
\node at (-1.2990,-0.7500) [circle,fill,inner sep=1.5pt]{};
\node at (-0.0000,-1.5000) [circle,fill,inner sep=1.5pt]{};
\node at (2.5981,1.5000) [circle,fill,inner sep=1.5pt]{};
\node at (2.5981,3.0000) [circle,fill,inner sep=1.5pt]{};
\node at (1.2990,3.7500) [circle,fill,inner sep=1.5pt]{};
\node at (-0.0000,3.0000) [circle,fill,inner sep=1.5pt]{};
\node at (0.0000,1.5000) [circle,fill,inner sep=1.5pt]{};
\node at (1.2990,0.7500) [circle,fill,inner sep=1.5pt]{};
\node at (3.8971,3.7500) [circle,fill,inner sep=1.5pt]{};
\node at (3.8971,5.2500) [circle,fill,inner sep=1.5pt]{};
\node at (2.5981,6.0000) [circle,fill,inner sep=1.5pt]{};
\node at (1.2990,5.2500) [circle,fill,inner sep=1.5pt]{};
\node at (1.2990,3.7500) [circle,fill,inner sep=1.5pt]{};
\node at (2.5981,3.0000) [circle,fill,inner sep=1.5pt]{};
\node at (3.8971,-0.7500) [circle,fill,inner sep=1.5pt]{};
\node at (3.8971,0.7500) [circle,fill,inner sep=1.5pt]{};
\node at (2.5981,1.5000) [circle,fill,inner sep=1.5pt]{};
\node at (1.2990,0.7500) [circle,fill,inner sep=1.5pt]{};
\node at (1.2990,-0.7500) [circle,fill,inner sep=1.5pt]{};
\node at (2.5981,-1.5000) [circle,fill,inner sep=1.5pt]{};
\node at (5.1962,1.5000) [circle,fill,inner sep=1.5pt]{};
\node at (5.1962,3.0000) [circle,fill,inner sep=1.5pt]{};
\node at (3.8971,3.7500) [circle,fill,inner sep=1.5pt]{};
\node at (2.5981,3.0000) [circle,fill,inner sep=1.5pt]{};
\node at (2.5981,1.5000) [circle,fill,inner sep=1.5pt]{};
\node at (3.8971,0.7500) [circle,fill,inner sep=1.5pt]{};
\node at (6.4952,3.7500) [circle,fill,inner sep=1.5pt]{};
\node at (6.4952,5.2500) [circle,fill,inner sep=1.5pt]{};
\node at (5.1962,6.0000) [circle,fill,inner sep=1.5pt]{};
\node at (3.8971,5.2500) [circle,fill,inner sep=1.5pt]{};
\node at (3.8971,3.7500) [circle,fill,inner sep=1.5pt]{};
\node at (5.1962,3.0000) [circle,fill,inner sep=1.5pt]{};
\node at (6.4952,-0.7500) [circle,fill,inner sep=1.5pt]{};
\node at (6.4952,0.7500) [circle,fill,inner sep=1.5pt]{};
\node at (5.1962,1.5000) [circle,fill,inner sep=1.5pt]{};
\node at (3.8971,0.7500) [circle,fill,inner sep=1.5pt]{};
\node at (3.8971,-0.7500) [circle,fill,inner sep=1.5pt]{};
\node at (5.1962,-1.5000) [circle,fill,inner sep=1.5pt]{};
\node at (7.7942,1.5000) [circle,fill,inner sep=1.5pt]{};
\node at (7.7942,3.0000) [circle,fill,inner sep=1.5pt]{};
\node at (6.4952,3.7500) [circle,fill,inner sep=1.5pt]{};
\node at (5.1962,3.0000) [circle,fill,inner sep=1.5pt]{};
\node at (5.1962,1.5000) [circle,fill,inner sep=1.5pt]{};
\node at (6.4952,0.7500) [circle,fill,inner sep=1.5pt]{};
\node at (9.0933,3.7500) [circle,fill,inner sep=1.5pt]{};
\node at (9.0933,5.2500) [circle,fill,inner sep=1.5pt]{};
\node at (7.7942,6.0000) [circle,fill,inner sep=1.5pt]{};
\node at (6.4952,5.2500) [circle,fill,inner sep=1.5pt]{};
\node at (6.4952,3.7500) [circle,fill,inner sep=1.5pt]{};
\node at (7.7942,3.0000) [circle,fill,inner sep=1.5pt]{};
\draw(-1.2990, -0.7500)--(-0.0000, -1.5000);
\draw(0.0000, 1.5000)--(1.2990, 0.7500);
\draw(1.2990, 3.7500)--(2.5981, 3.0000);
\draw(2.5981, 6.0000)--(3.8971, 5.2500);
\draw(1.2990, -0.7500)--(2.5981, -1.5000);
\draw(2.5981, 1.5000)--(3.8971, 0.7500);
\draw(3.8971, 3.7500)--(5.1962, 3.0000);
\draw(5.1962, 6.0000)--(6.4952, 5.2500);
\draw(3.8971, -0.7500)--(5.1962, -1.5000);
\draw(5.1962, 1.5000)--(6.4952, 0.7500);
\draw(6.4952, 3.7500)--(7.7942, 3.0000);
\draw(7.7942, 6.0000)--(9.0933, 5.2500);
\draw(-1.2990, -0.7500)--(-1.2990, 0.7500);
\draw(0.0000, 1.5000)--(-0.0000, 3.0000);
\draw(1.2990, 3.7500)--(1.2990, 5.2500);
\draw(1.2990, -0.7500)--(1.2990, 0.7500);
\draw(2.5981, 1.5000)--(2.5981, 3.0000);
\draw(3.8971, 3.7500)--(3.8971, 5.2500);
\draw(3.8971, -0.7500)--(3.8971, 0.7500);
\draw(5.1962, 1.5000)--(5.1962, 3.0000);
\draw(6.4952, 3.7500)--(6.4952, 5.2500);
\draw(6.4952, -0.7500)--(6.4952, 0.7500);
\draw(7.7942, 1.5000)--(7.7942, 3.0000);
\draw(9.0933, 3.7500)--(9.0933, 5.2500);
\draw(-1.2990, 0.7500)--(-0.0000, 1.5000);
\draw(-0.0000, 3.0000)--(1.2990, 3.7500);
\draw(1.2990, 5.2500)--(2.5981, 6.0000);
\draw(-0.0000, -1.5000)--(1.2990, -0.7500);
\draw(1.2990, 0.7500)--(2.5981, 1.5000);
\draw(2.5981, 3.0000)--(3.8971, 3.7500);
\draw(3.8971, 5.2500)--(5.1962, 6.0000);
\draw(2.5981, -1.5000)--(3.8971, -0.7500);
\draw(3.8971, 0.7500)--(5.1962, 1.5000);
\draw(5.1962, 3.0000)--(6.4952, 3.7500);
\draw(6.4952, 5.2500)--(7.7942, 6.0000);
\draw(5.1962, -1.5000)--(6.4952, -0.7500);
\draw(6.4952, 0.7500)--(7.7942, 1.5000);
\draw(7.7942, 3.0000)--(9.0933, 3.7500);
\end{tikzpicture}
}
\caption{The slanted grid $\mathcal B_\ell$ with $\ell=3$.}.
\label{fig:hexslant}
\end{center}
\end{figure}
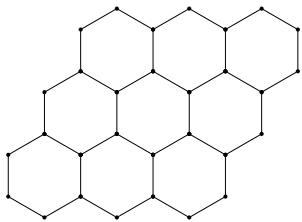
\vspace{-.5cm}

The finite-size criterion on the honeycomb lattice then reads as follows.

\begin{thm}[Main result on the honeycomb lattice]\label{thm:mainhoney}
Let $\ell \geq 3, L > 2\ell$. Then,
\begin{equation} \label{eqn:hexagongap}
\gamma_L^{\mathrm{per}} \geq \frac{1}{2}\left( \gamma_\ell -t_\ell\right)
\end{equation}
where 
\beq
t_\ell= \frac{228}{55\ell^2} + \frac{108}{\ell^3},\qquad \ell\geq 10.
\eeq
For $3\leq \ell\leq 9$, the thresholds $t_\ell$ are displayed in Table \ref{tab:table4}.
\end{thm}

\begin{table}[h]
  \begin{center}
    \begin{tabular}{c|cccccccc|}
        $\ell$ & 3 & 4&5&6&7&8&9\\ \hline
        $t_\ell$ & 0.246 & 0.161 & 0.113 & 0.084 & 0.065 &0.051 & 0.042 \\ \hline
     $\ell^2t_\ell$ & 2.212 & 2.576 & 2.824 & 3.003 & 3.140 & 3.247   & 3.333
    \end{tabular}
    \caption{Thresholds $t_\ell $  on the honeycomb lattice for small values of $\ell$.}
        \label{tab:table3}
  \end{center}
\end{table}

\be{rmk}
\be{enumerate}[label=(\roman*)]
\item

Analogously to the discussion in Remark \ref{rmk:main} (i), the restriction to $ \ell\geq 10$ arises from a desire to have a reasonable constant in the error term (here $108$), while the numbers obtained for $3\leq \ell\leq 9$ arise through a more targeted optimization procedure of the parameters chosen to define the weighting scheme in the proof.

\item
Concerning the best possible threshold on the honeycomb lattice, we can again consider spin wave trial states in the Heisenberg ferromagnet. We performed numerics that then predict a gap scaling  $\approx \frac{0.9}{\ell^2}$ which, if sharp, would leave the possibility of further improvement of the constant $\frac{228}{55}$.

\e{enumerate}
\e{rmk}

\subsection{Finite-size criterion for the triangular lattice}
We obtain an analogous result for the triangular lattice. We let $\mathbb T_L$ denote the triangular lattice wrapped on an $L\times L$ torus and write $H^{\mathrm{per}}_{\mathbb T_L}$ for the associated Hamiltonian, which we assume to be frustration-free, and $\gamma_L^{\mathrm{per}}$ for its spectral gap. The relevant finite subsystems are $\ell \times \ell$ ``slanted" grids of triangles as shown in Figure \ref{fig:B3tri}. The associated subsystem Hamiltonians are denoted by $H_{\mathcal{B}_\ell}$ with gaps $\gamma_\ell$. 

\vspace{.5cm}
\begin{figure}[h]
\begin{center}
\resizebox{4cm}{!}{
 \begin{tikzpicture}
\node at (2.5981,-0.7500) [circle,fill,inner sep=1.5pt]{};
\node at (1.2990,1.5000) [circle,fill,inner sep=1.5pt]{};
\node at (0.0000,-0.7500) [circle,fill,inner sep=1.5pt]{};
\node at (2.5981,-0.7500) [circle,fill,inner sep=1.5pt]{};
\node at (1.2990,1.5000) [circle,fill,inner sep=1.5pt]{};
\node at (-0.0000,-0.7500) [circle,fill,inner sep=1.5pt]{};
\node at (3.8971,1.5000) [circle,fill,inner sep=1.5pt]{};
\node at (2.5981,3.7500) [circle,fill,inner sep=1.5pt]{};
\node at (1.2990,1.5000) [circle,fill,inner sep=1.5pt]{};
\node at (3.8971,1.5000) [circle,fill,inner sep=1.5pt]{};
\node at (2.5981,3.7500) [circle,fill,inner sep=1.5pt]{};
\node at (1.2990,1.5000) [circle,fill,inner sep=1.5pt]{};
\node at (5.1962,3.7500) [circle,fill,inner sep=1.5pt]{};
\node at (3.8971,6.0000) [circle,fill,inner sep=1.5pt]{};
\node at (2.5981,3.7500) [circle,fill,inner sep=1.5pt]{};
\node at (5.1962,3.7500) [circle,fill,inner sep=1.5pt]{};
\node at (3.8971,6.0000) [circle,fill,inner sep=1.5pt]{};
\node at (2.5981,3.7500) [circle,fill,inner sep=1.5pt]{};
\node at (5.1962,-0.7500) [circle,fill,inner sep=1.5pt]{};
\node at (3.8971,1.5000) [circle,fill,inner sep=1.5pt]{};
\node at (2.5981,-0.7500) [circle,fill,inner sep=1.5pt]{};
\node at (5.1962,-0.7500) [circle,fill,inner sep=1.5pt]{};
\node at (3.8971,1.5000) [circle,fill,inner sep=1.5pt]{};
\node at (2.5981,-0.7500) [circle,fill,inner sep=1.5pt]{};
\node at (6.4952,1.5000) [circle,fill,inner sep=1.5pt]{};
\node at (5.1962,3.7500) [circle,fill,inner sep=1.5pt]{};
\node at (3.8971,1.5000) [circle,fill,inner sep=1.5pt]{};
\node at (6.4952,1.5000) [circle,fill,inner sep=1.5pt]{};
\node at (5.1962,3.7500) [circle,fill,inner sep=1.5pt]{};
\node at (3.8971,1.5000) [circle,fill,inner sep=1.5pt]{};
\node at (7.7942,3.7500) [circle,fill,inner sep=1.5pt]{};
\node at (6.4952,6.0000) [circle,fill,inner sep=1.5pt]{};
\node at (5.1962,3.7500) [circle,fill,inner sep=1.5pt]{};
\node at (7.7942,3.7500) [circle,fill,inner sep=1.5pt]{};
\node at (6.4952,6.0000) [circle,fill,inner sep=1.5pt]{};
\node at (5.1962,3.7500) [circle,fill,inner sep=1.5pt]{};
\node at (7.7942,-0.7500) [circle,fill,inner sep=1.5pt]{};
\node at (6.4952,1.5000) [circle,fill,inner sep=1.5pt]{};
\node at (5.1962,-0.7500) [circle,fill,inner sep=1.5pt]{};
\node at (7.7942,-0.7500) [circle,fill,inner sep=1.5pt]{};
\node at (6.4952,1.5000) [circle,fill,inner sep=1.5pt]{};
\node at (5.1962,-0.7500) [circle,fill,inner sep=1.5pt]{};
\node at (9.0933,1.5000) [circle,fill,inner sep=1.5pt]{};
\node at (7.7942,3.7500) [circle,fill,inner sep=1.5pt]{};
\node at (6.4952,1.5000) [circle,fill,inner sep=1.5pt]{};
\node at (9.0933,1.5000) [circle,fill,inner sep=1.5pt]{};
\node at (7.7942,3.7500) [circle,fill,inner sep=1.5pt]{};
\node at (6.4952,1.5000) [circle,fill,inner sep=1.5pt]{};
\node at (10.3923,3.7500) [circle,fill,inner sep=1.5pt]{};
\node at (9.0933,6.0000) [circle,fill,inner sep=1.5pt]{};
\node at (7.7942,3.7500) [circle,fill,inner sep=1.5pt]{};
\node at (10.3923,3.7500) [circle,fill,inner sep=1.5pt]{};
\node at (9.0933,6.0000) [circle,fill,inner sep=1.5pt]{};
\node at (7.7942,3.7500) [circle,fill,inner sep=1.5pt]{};
\node at (11.6913,6.0000) [circle,fill,inner sep=1.5pt]{};
\draw(0.0000, -0.7500)--(2.5981, -0.7500);
\draw(1.2990, 1.5000)--(3.8971, 1.5000);
\draw(2.5981, 3.7500)--(5.1962, 3.7500);
\draw(3.8971, 6.0000)--(6.4952, 6.0000);
\draw(2.5981, -0.7500)--(5.1962, -0.7500);
\draw(3.8971, 1.5000)--(6.4952, 1.5000);
\draw(5.1962, 3.7500)--(7.7942, 3.7500);
\draw(6.4952, 6.0000)--(9.0933, 6.0000);
\draw(5.1962, -0.7500)--(7.7942, -0.7500);
\draw(6.4952, 1.5000)--(9.0933, 1.5000);
\draw(7.7942, 3.7500)--(10.3923, 3.7500);
\draw(9.0933, 6.0000)--(11.6913, 6.0000);
\draw(1.2990, 1.5000)--(0.0000, -0.7500);
\draw(2.5981, 3.7500)--(1.2990, 1.5000);
\draw(3.8971, 6.0000)--(2.5981, 3.7500);
\draw(3.8971, 1.5000)--(2.5981, -0.7500);
\draw(5.1962, 3.7500)--(3.8971, 1.5000);
\draw(6.4952, 6.0000)--(5.1962, 3.7500);
\draw(6.4952, 1.5000)--(5.1962, -0.7500);
\draw(7.7942, 3.7500)--(6.4952, 1.5000);
\draw(9.0933, 6.0000)--(7.7942, 3.7500);
\draw(9.0933, 1.5000)--(7.7942, -0.7500);
\draw(10.3923, 3.7500)--(9.0933, 1.5000);
\draw(11.6913, 6.0000)--(10.3923, 3.7500);
\draw(1.2990, 1.5000)--(2.5981, -0.7500);
\draw(2.5981, 3.7500)--(3.8971, 1.5000);
\draw(3.8971, 6.0000)--(5.1962, 3.7500);
\draw(3.8971, 1.5000)--(5.1962, -0.7500);
\draw(5.1962, 3.7500)--(6.4952, 1.5000);
\draw(6.4952, 6.0000)--(7.7942, 3.7500);
\draw(6.4952, 1.5000)--(7.7942, -0.7500);
\draw(7.7942, 3.7500)--(9.0933, 1.5000);
\draw(9.0933, 6.0000)--(10.3923, 3.7500);
\end{tikzpicture}
}
\caption{The slanted triangular grid $\mathcal B_\ell$ with $\ell=3$.}.
\label{fig:B3tri}
\end{center}
\end{figure}
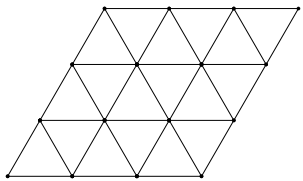
\vspace{-.5cm}

The finite-size criterion on the triangular lattice then reads as follows.

\begin{thm}[Main result on triangular lattice]
\label{thm:maintri}
\label{eqn:trianglegap}
Let $\ell \geq 3, L > 2\ell$. Then,
\begin{equation}
\gamma_L^{\mathrm{per}} \geq \frac{1}{2}\left(\gamma_\ell - t_\ell \right).
\end{equation}
where 
\beq
t_\ell =\frac{144}{5\ell^2} + \frac{432}{\ell^3},\qquad \ell\geq 10.
\eeq
For $3\leq \ell\leq 9$, the thresholds $t_\ell$ are displayed in Table \ref{tab:table4}.
\end{thm}


\begin{table}[!ht]
  \begin{center}
    \begin{tabular}{c|cccccccc|} 
        $\ell$ & 3 & 4&5&6&7&8&9\\ \hline
        $t_\ell$ & 1.318 & 0.872 & 0.628 & 0.476 & 0.373 &0.301 & 0.248 \\ \hline
     $\ell^2 t_\ell$ & 11.861 & 13.946 & 15.700 & 17.113 & 18.263 & 19.213   &  20.010
    \end{tabular}
        \caption{Thresholds for the finite-size criterion on the triangular lattice for small values of $\ell$.}
            \label{tab:table4}
  \end{center}
\end{table}

We remark that it is no coincidence that the constants appearing in Theorem \ref{thm:maintri} are comparatively worse than in the previous two theorems. The underlying reason is that the triangular lattice has a relatively large number of bonds incident at each site and thus has a large number of locally competing interactions. This makes it comparatively difficult to control the sum of anticommutators corresponding to adjacent edges that arise in the proof from squaring the weighted Hamiltonian. As a benchmark for assessing the sharpness of the finite-size criterion, we note using spin wave trial states for the ferromagnetic Heisenberg model with $S = \frac{1}{2}$ together with numerical computation of eigenvalues of the discrete Laplacian suggest a gap scaling of $\frac{5.3}{\ell^2}$ for large $\ell$, which, if sharp, would still leave room for improvement of the constant $\frac{144}{5}$ that we find on the triangular lattice.

\section{Proof of Theorem \ref{thm:main} on Euclidean lattices}
\subsection{Construction of weighted subsystem Hamiltonians}\label{ssect:coefficients}
Let $D\geq 2$. We will establish \eqref{eqn:thm1} by studying weighted Hamiltonians on $\mathcal{B}_\ell$. This is inspired by the approach of \cite{GM}, though, the weighting scheme is completely different here. The weighted Hamiltonian $W_{\mathcal{B}_\ell}$, is defined as follows. We use two families of coefficients $c_0, \dotsc c_{\ell-1}$ and $d_0, \dotsc d_{\ell}$ (to be further specified later) satisfying the following requirements:
\begin{itemize}
    \item[(i)] The $c_i$ and $d_i$ are positive.
    \item[(ii)] The $c_i,$ and $d_i$ are symmetric around the midpoint, i.e., $c_i = c_{\ell-1-i}$ and $d_i = d_{\ell-i}$. 
    \item[(iii)]The $c_i$ and $d_i$ are increasing up to the midpoint. That is, $c_i \leq c_{i+1}$ for $0\leq i \leq \lfloor \frac{\ell-1}{2}\rfloor$, and $d_i \leq d_{i+1}$ for $0 \leq i \leq \lfloor \frac{\ell}{2}\rfloor$. 
\end{itemize}
Then, we define the weighted subsystem Hamiltonians as
\beq\label{eq:WBgrid}
W_{\mathcal{B}_\ell} = \sum_{1 \leq j \leq D} \sum_{\substack{x\in \mathcal{B}_\ell: \\ 0 \leq x_j \leq \ell-1}} c_{x_j} \left(\prod_{\substack{1\leq k \leq D: \\ k \neq j}} d_{x_k} \right)h_{x,e_j}.
\eeq
Note that these subsystem Hamiltonians have open boundary conditions.

Figure \ref{fig:WB3grid} depicts the distribution of edge weights that create $W_{\mathcal{B}_3}$ in two dimensions. Notice that conditions (i)-(iii) imply that $c_0=c_2\leq c_1$ and $d_0=d_3\leq d_1=d_2$, so our weighting scheme implements the general principle that having large edge weights near the center and comparatively smaller edge weights near the boundary leads to the best thresholds. (This principle can be heuristically understood from the proof of the finite-size criterion leads where the goal ist to efficiently cover the lattice through translates of  $W_{\mathcal{B}_3}$.) The particular scheme that we propose in \eqref{eq:WBgrid} turns out to implement this general principle in a very efficient manner that allows for convenient calculation especially in higher dimensions. 

\vspace{.5cm}
\begin{figure}[h]
\begin{center}
\resizebox{5cm}{!}{
\begin{tikzpicture}
\node at (1.0607,-1.0607) [circle,fill,inner sep=1.5pt]{};
\node at (1.0607,1.0607) [circle,fill,inner sep=1.5pt]{};
\node at (-1.0607,1.0607) [circle,fill,inner sep=1.5pt]{};
\node at (-1.0607,-1.0607) [circle,fill,inner sep=1.5pt]{};
\node at (1.0607,1.0607) [circle,fill,inner sep=1.5pt]{};
\node at (1.0607,3.1820) [circle,fill,inner sep=1.5pt]{};
\node at (-1.0607,3.1820) [circle,fill,inner sep=1.5pt]{};
\node at (-1.0607,1.0607) [circle,fill,inner sep=1.5pt]{};
\node at (1.0607,3.1820) [circle,fill,inner sep=1.5pt]{};
\node at (1.0607,5.3033) [circle,fill,inner sep=1.5pt]{};
\node at (-1.0607,5.3033) [circle,fill,inner sep=1.5pt]{};
\node at (-1.0607,3.1820) [circle,fill,inner sep=1.5pt]{};
\node at (3.1820,-1.0607) [circle,fill,inner sep=1.5pt]{};
\node at (3.1820,1.0607) [circle,fill,inner sep=1.5pt]{};
\node at (1.0607,1.0607) [circle,fill,inner sep=1.5pt]{};
\node at (1.0607,-1.0607) [circle,fill,inner sep=1.5pt]{};
\node at (3.1820,1.0607) [circle,fill,inner sep=1.5pt]{};
\node at (3.1820,3.1820) [circle,fill,inner sep=1.5pt]{};
\node at (1.0607,3.1820) [circle,fill,inner sep=1.5pt]{};
\node at (1.0607,1.0607) [circle,fill,inner sep=1.5pt]{};
\node at (3.1820,3.1820) [circle,fill,inner sep=1.5pt]{};
\node at (3.1820,5.3033) [circle,fill,inner sep=1.5pt]{};
\node at (1.0607,5.3033) [circle,fill,inner sep=1.5pt]{};
\node at (1.0607,3.1820) [circle,fill,inner sep=1.5pt]{};
\node at (5.3033,-1.0607) [circle,fill,inner sep=1.5pt]{};
\node at (5.3033,1.0607) [circle,fill,inner sep=1.5pt]{};
\node at (3.1820,1.0607) [circle,fill,inner sep=1.5pt]{};
\node at (3.1820,-1.0607) [circle,fill,inner sep=1.5pt]{};
\node at (5.3033,1.0607) [circle,fill,inner sep=1.5pt]{};
\node at (5.3033,3.1820) [circle,fill,inner sep=1.5pt]{};
\node at (3.1820,3.1820) [circle,fill,inner sep=1.5pt]{};
\node at (3.1820,1.0607) [circle,fill,inner sep=1.5pt]{};
\node at (5.3033,3.1820) [circle,fill,inner sep=1.5pt]{};
\node at (5.3033,5.3033) [circle,fill,inner sep=1.5pt]{};
\node at (3.1820,5.3033) [circle,fill,inner sep=1.5pt]{};
\node at (3.1820,3.1820) [circle,fill,inner sep=1.5pt]{};
\draw(-1.0607, 1.0607)--(-1.0607, -1.0607) node[pos=.5]{$d_0c_0$};
\draw(-1.0607, 3.1820)--(-1.0607, 1.0607) node[pos=.5]{$d_0c_1$};
\draw(-1.0607, 5.3033)--(-1.0607, 3.1820) node[pos=.5]{$d_0c_2$};
\draw(1.0607, 1.0607)--(1.0607, -1.0607) node[pos=.5]{$d_1c_0$};
\draw(1.0607, 3.1820)--(1.0607, 1.0607) node[pos=.5]{$d_1c_1$};
\draw(1.0607, 5.3033)--(1.0607, 3.1820) node[pos=.5]{$d_1c_2$};
\draw(3.1820, 1.0607)--(3.1820, -1.0607) node[pos=.5]{$d_2c_0$};
\draw(3.1820, 3.1820)--(3.1820, 1.0607) node[pos=.5]{$d_2c_1$};
\draw(3.1820, 5.3033)--(3.1820, 3.1820) node[pos=.5]{$d_2c_2$};
\draw(5.3033, 1.0607)--(5.3033, -1.0607) node[pos=.5]{$d_3c_0$};
\draw(5.3033, 3.1820)--(5.3033, 1.0607) node[pos=.5]{$d_3c_1$};
\draw(5.3033, 5.3033)--(5.3033, 3.1820) node[pos=.5]{$d_3c_2$};
\draw(-1.0607, -1.0607)--(1.0607, -1.0607) node[pos=.5]{$c_0d_0$};
\draw(-1.0607, 1.0607)--(1.0607, 1.0607) node[pos=.5]{$c_0d_1$};
\draw(-1.0607, 3.1820)--(1.0607, 3.1820) node[pos=.5]{$c_0d_2$};
\draw(-1.0607, 5.3033)--(1.0607, 5.3033) node[pos=.5]{$c_0d_3$};
\draw(1.0607, -1.0607)--(3.1820, -1.0607) node[pos=.5]{$c_1d_0$};
\draw(1.0607, 1.0607)--(3.1820, 1.0607) node[pos=.5]{$c_1d_1$};
\draw(1.0607, 3.1820)--(3.1820, 3.1820) node[pos=.5]{$c_1d_2$};
\draw(1.0607, 5.3033)--(3.1820, 5.3033) node[pos=.5]{$c_1d_3$};
\draw(3.1820, -1.0607)--(5.3033, -1.0607) node[pos=.5]{$c_2d_0$};
\draw(3.1820, 1.0607)--(5.3033, 1.0607) node[pos=.5]{$c_2d_1$};
\draw(3.1820, 3.1820)--(5.3033, 3.1820) node[pos=.5]{$c_2d_2$};
\draw(3.1820, 5.3033)--(5.3033, 5.3033) node[pos=.5]{$c_2d_3$};
\end{tikzpicture}
}
\caption{The weighted subsystem Hamiltonian $W_{\mathcal{B}_3}$ on the 2D square grid. Edges are labeled by their weights. Notice that edges in the center have a comparatively higher weight compared to edges along the boundary, which in turn have a relatively higher weight compared to edges at a corner.}
\label{fig:WB3grid}
\end{center}
\end{figure}
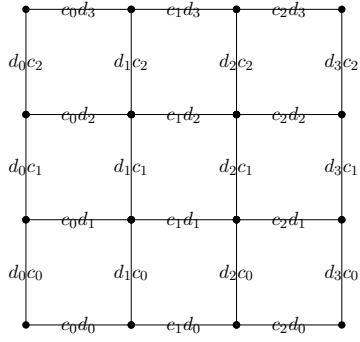
\vspace{-.5cm}

To understand the product nature of the scheme, observe that the two-dimensional $W_{\mathcal{B}_3}$ depicted in Figure \ref{fig:WB3grid} arises as the Hadamard (pointwise) product from two simpler weighted Hamiltonians displayed in Figure \ref{fig:WB3gridproduct}. 
This idea of using products as simpler building blocks for realizing \eqref{eq:WBgrid} extends well to higher dimensions. Then the independence-like properties of products are what underlies the calculational convenience of this weighting scheme in higher dimensions. To summarize, formula \eqref{eq:WBgrid} is the key theoretical underpinning of the improved thresholds.

\vspace{.5cm}
\begin{figure}
\begin{center}
\resizebox{8cm}{!}{
\begin{minipage}{.33\textwidth}
\begin{tikzpicture}
\node at (0.7071,-0.7071) [circle,fill,inner sep=1.5pt]{};
\node at (0.7071,0.7071) [circle,fill,inner sep=1.5pt]{};
\node at (-0.7071,0.7071) [circle,fill,inner sep=1.5pt]{};
\node at (-0.7071,-0.7071) [circle,fill,inner sep=1.5pt]{};
\node at (0.7071,0.7071) [circle,fill,inner sep=1.5pt]{};
\node at (0.7071,2.1213) [circle,fill,inner sep=1.5pt]{};
\node at (-0.7071,2.1213) [circle,fill,inner sep=1.5pt]{};
\node at (-0.7071,0.7071) [circle,fill,inner sep=1.5pt]{};
\node at (0.7071,2.1213) [circle,fill,inner sep=1.5pt]{};
\node at (0.7071,3.5355) [circle,fill,inner sep=1.5pt]{};
\node at (-0.7071,3.5355) [circle,fill,inner sep=1.5pt]{};
\node at (-0.7071,2.1213) [circle,fill,inner sep=1.5pt]{};
\node at (2.1213,-0.7071) [circle,fill,inner sep=1.5pt]{};
\node at (2.1213,0.7071) [circle,fill,inner sep=1.5pt]{};
\node at (0.7071,0.7071) [circle,fill,inner sep=1.5pt]{};
\node at (0.7071,-0.7071) [circle,fill,inner sep=1.5pt]{};
\node at (2.1213,0.7071) [circle,fill,inner sep=1.5pt]{};
\node at (2.1213,2.1213) [circle,fill,inner sep=1.5pt]{};
\node at (0.7071,2.1213) [circle,fill,inner sep=1.5pt]{};
\node at (0.7071,0.7071) [circle,fill,inner sep=1.5pt]{};
\node at (2.1213,2.1213) [circle,fill,inner sep=1.5pt]{};
\node at (2.1213,3.5355) [circle,fill,inner sep=1.5pt]{};
\node at (0.7071,3.5355) [circle,fill,inner sep=1.5pt]{};
\node at (0.7071,2.1213) [circle,fill,inner sep=1.5pt]{};
\node at (3.5355,-0.7071) [circle,fill,inner sep=1.5pt]{};
\node at (3.5355,0.7071) [circle,fill,inner sep=1.5pt]{};
\node at (2.1213,0.7071) [circle,fill,inner sep=1.5pt]{};
\node at (2.1213,-0.7071) [circle,fill,inner sep=1.5pt]{};
\node at (3.5355,0.7071) [circle,fill,inner sep=1.5pt]{};
\node at (3.5355,2.1213) [circle,fill,inner sep=1.5pt]{};
\node at (2.1213,2.1213) [circle,fill,inner sep=1.5pt]{};
\node at (2.1213,0.7071) [circle,fill,inner sep=1.5pt]{};
\node at (3.5355,2.1213) [circle,fill,inner sep=1.5pt]{};
\node at (3.5355,3.5355) [circle,fill,inner sep=1.5pt]{};
\node at (2.1213,3.5355) [circle,fill,inner sep=1.5pt]{};
\node at (2.1213,2.1213) [circle,fill,inner sep=1.5pt]{};
\draw(-0.7071, 0.7071)--(-0.7071, -0.7071) node[pos=.5]{$d_0$};
\draw(-0.7071, 2.1213)--(-0.7071, 0.7071) node[pos=.5]{$d_0$};
\draw(-0.7071, 3.5355)--(-0.7071, 2.1213) node[pos=.5]{$d_0$};
\draw(0.7071, 0.7071)--(0.7071, -0.7071) node[pos=.5]{$d_1$};
\draw(0.7071, 2.1213)--(0.7071, 0.7071) node[pos=.5]{$d_1$};
\draw(0.7071, 3.5355)--(0.7071, 2.1213) node[pos=.5]{$d_1$};
\draw(2.1213, 0.7071)--(2.1213, -0.7071) node[pos=.5]{$d_2$};
\draw(2.1213, 2.1213)--(2.1213, 0.7071) node[pos=.5]{$d_2$};
\draw(2.1213, 3.5355)--(2.1213, 2.1213) node[pos=.5]{$d_2$};
\draw(3.5355, 0.7071)--(3.5355, -0.7071) node[pos=.5]{$d_3$};
\draw(3.5355, 2.1213)--(3.5355, 0.7071) node[pos=.5]{$d_3$};
\draw(3.5355, 3.5355)--(3.5355, 2.1213) node[pos=.5]{$d_3$};
\draw(-0.7071, -0.7071)--(0.7071, -0.7071) node[pos=.5]{$c_0$};
\draw(-0.7071, 0.7071)--(0.7071, 0.7071) node[pos=.5]{$c_0$};
\draw(-0.7071, 2.1213)--(0.7071, 2.1213) node[pos=.5]{$c_0$};
\draw(-0.7071, 3.5355)--(0.7071, 3.5355) node[pos=.5]{$c_0$};
\draw(0.7071, -0.7071)--(2.1213, -0.7071) node[pos=.5]{$c_1$};
\draw(0.7071, 0.7071)--(2.1213, 0.7071) node[pos=.5]{$c_1$};
\draw(0.7071, 2.1213)--(2.1213, 2.1213) node[pos=.5]{$c_1$};
\draw(0.7071, 3.5355)--(2.1213, 3.5355) node[pos=.5]{$c_1$};
\draw(2.1213, -0.7071)--(3.5355, -0.7071) node[pos=.5]{$c_2$};
\draw(2.1213, 0.7071)--(3.5355, 0.7071) node[pos=.5]{$c_2$};
\draw(2.1213, 2.1213)--(3.5355, 2.1213) node[pos=.5]{$c_2$};
\draw(2.1213, 3.5355)--(3.5355, 3.5355) node[pos=.5]{$c_2$};
\end{tikzpicture}
\end{minipage}
$\times$\quad\quad
\begin{minipage}{.33\textwidth}
\begin{tikzpicture}
\node at (0.7071,-0.7071) [circle,fill,inner sep=1.5pt]{};
\node at (0.7071,0.7071) [circle,fill,inner sep=1.5pt]{};
\node at (-0.7071,0.7071) [circle,fill,inner sep=1.5pt]{};
\node at (-0.7071,-0.7071) [circle,fill,inner sep=1.5pt]{};
\node at (0.7071,0.7071) [circle,fill,inner sep=1.5pt]{};
\node at (0.7071,2.1213) [circle,fill,inner sep=1.5pt]{};
\node at (-0.7071,2.1213) [circle,fill,inner sep=1.5pt]{};
\node at (-0.7071,0.7071) [circle,fill,inner sep=1.5pt]{};
\node at (0.7071,2.1213) [circle,fill,inner sep=1.5pt]{};
\node at (0.7071,3.5355) [circle,fill,inner sep=1.5pt]{};
\node at (-0.7071,3.5355) [circle,fill,inner sep=1.5pt]{};
\node at (-0.7071,2.1213) [circle,fill,inner sep=1.5pt]{};
\node at (2.1213,-0.7071) [circle,fill,inner sep=1.5pt]{};
\node at (2.1213,0.7071) [circle,fill,inner sep=1.5pt]{};
\node at (0.7071,0.7071) [circle,fill,inner sep=1.5pt]{};
\node at (0.7071,-0.7071) [circle,fill,inner sep=1.5pt]{};
\node at (2.1213,0.7071) [circle,fill,inner sep=1.5pt]{};
\node at (2.1213,2.1213) [circle,fill,inner sep=1.5pt]{};
\node at (0.7071,2.1213) [circle,fill,inner sep=1.5pt]{};
\node at (0.7071,0.7071) [circle,fill,inner sep=1.5pt]{};
\node at (2.1213,2.1213) [circle,fill,inner sep=1.5pt]{};
\node at (2.1213,3.5355) [circle,fill,inner sep=1.5pt]{};
\node at (0.7071,3.5355) [circle,fill,inner sep=1.5pt]{};
\node at (0.7071,2.1213) [circle,fill,inner sep=1.5pt]{};
\node at (3.5355,-0.7071) [circle,fill,inner sep=1.5pt]{};
\node at (3.5355,0.7071) [circle,fill,inner sep=1.5pt]{};
\node at (2.1213,0.7071) [circle,fill,inner sep=1.5pt]{};
\node at (2.1213,-0.7071) [circle,fill,inner sep=1.5pt]{};
\node at (3.5355,0.7071) [circle,fill,inner sep=1.5pt]{};
\node at (3.5355,2.1213) [circle,fill,inner sep=1.5pt]{};
\node at (2.1213,2.1213) [circle,fill,inner sep=1.5pt]{};
\node at (2.1213,0.7071) [circle,fill,inner sep=1.5pt]{};
\node at (3.5355,2.1213) [circle,fill,inner sep=1.5pt]{};
\node at (3.5355,3.5355) [circle,fill,inner sep=1.5pt]{};
\node at (2.1213,3.5355) [circle,fill,inner sep=1.5pt]{};
\node at (2.1213,2.1213) [circle,fill,inner sep=1.5pt]{};
\draw(-0.7071, 0.7071)--(-0.7071, -0.7071) node[pos=.5]{$c_0$};
\draw(-0.7071, 2.1213)--(-0.7071, 0.7071) node[pos=.5]{$c_1$};
\draw(-0.7071, 3.5355)--(-0.7071, 2.1213) node[pos=.5]{$c_2$};
\draw(0.7071, 0.7071)--(0.7071, -0.7071) node[pos=.5]{$c_0$};
\draw(0.7071, 2.1213)--(0.7071, 0.7071) node[pos=.5]{$c_1$};
\draw(0.7071, 3.5355)--(0.7071, 2.1213) node[pos=.5]{$c_2$};
\draw(2.1213, 0.7071)--(2.1213, -0.7071) node[pos=.5]{$c_0$};
\draw(2.1213, 2.1213)--(2.1213, 0.7071) node[pos=.5]{$c_1$};
\draw(2.1213, 3.5355)--(2.1213, 2.1213) node[pos=.5]{$c_2$};
\draw(3.5355, 0.7071)--(3.5355, -0.7071) node[pos=.5]{$c_0$};
\draw(3.5355, 2.1213)--(3.5355, 0.7071) node[pos=.5]{$c_1$};
\draw(3.5355, 3.5355)--(3.5355, 2.1213) node[pos=.5]{$c_2$};
\draw(-0.7071, -0.7071)--(0.7071, -0.7071) node[pos=.5]{$d_0$};
\draw(-0.7071, 0.7071)--(0.7071, 0.7071) node[pos=.5]{$d_1$};
\draw(-0.7071, 2.1213)--(0.7071, 2.1213) node[pos=.5]{$d_2$};
\draw(-0.7071, 3.5355)--(0.7071, 3.5355) node[pos=.5]{$d_3$};
\draw(0.7071, -0.7071)--(2.1213, -0.7071) node[pos=.5]{$d_0$};
\draw(0.7071, 0.7071)--(2.1213, 0.7071) node[pos=.5]{$d_1$};
\draw(0.7071, 2.1213)--(2.1213, 2.1213) node[pos=.5]{$d_2$};
\draw(0.7071, 3.5355)--(2.1213, 3.5355) node[pos=.5]{$d_3$};
\draw(2.1213, -0.7071)--(3.5355, -0.7071) node[pos=.5]{$d_0$};
\draw(2.1213, 0.7071)--(3.5355, 0.7071) node[pos=.5]{$d_1$};
\draw(2.1213, 2.1213)--(3.5355, 2.1213) node[pos=.5]{$d_2$};
\draw(2.1213, 3.5355)--(3.5355, 3.5355) node[pos=.5]{$d_3$};
\end{tikzpicture}    
    \end{minipage}
}
\caption{The weighted subsystem Hamiltonian $W_{\mathcal{B}_3}$ displayed in Figure \ref{fig:WB3grid} arises as the pointwise product of these two simpler weighting schemes.}
\label{fig:WB3gridproduct}
\end{center}
\end{figure}
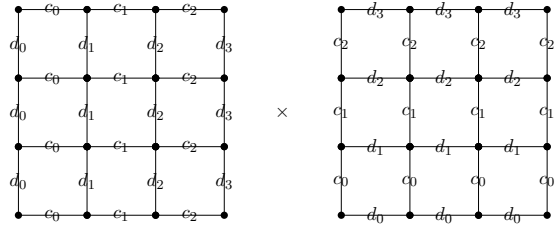
\vspace{-.5cm}

These weighted Hamiltonians will be used to prove the following general gap bound.

\be{prop}[General gap bound on the Euclidean lattice]\label{prop:keygrid} For $\ell\geq 2$, let the coefficients $c_0, \dotsc c_{\ell-1}$ and $d_0, \dotsc d_{\ell}$ satisfy requirements (i)-(iii) above. Under the assumptions of Theorem \ref{thm:main}, we have
\beq\label{eq:propclaim}
  \gamma_L^{\mathrm{per}} \geq \frac{K_4}{K_3}\left(\gamma_\ell - \frac{K_0+K_3-2K_1  }{K_4} \right),
  \eeq
  provided that $K_3\geq \max\{K_1,K_2\}$, where we introduced the effective constants
  \beq\label{eq:Kdefns}
\begin{aligned}
K_0 &= \left( \sum_{i=0}^{\ell -1} c_i^2\right) \left( \sum_{i=0}^\ell  d_i^2\right)^{D-1}\\
    K_1 &=\left(\sum_{i=0}^{\ell -1}c_i^2\right)\left(\sum_{i=0}^{\ell -1}d_id_{i+1} \right)\left(\sum_{i=0}^\ell  d_i^2 \right)^{D-2} \\
K_2 &= \left(\sum_{i=0}^{\ell -2}c_ic_{i+1}\right)\left(\sum_{i=0}^\ell  d_i^2 \right)^{D-1} \\
K_3 &= \left( \sum_{i=0}^{\ell -1}c_id_i\right)^2 \left(\sum_{i=0}^{\ell } d_i^2 \right)^{D-2} \\
K_4&= \frac{1}{\ell (\ell +1)^{D-1}} \left(\sum_{i=0}^{\ell -1}c_i \right)^2\left( \sum_{i=0}^\ell  d_i\right)^{2(D-1)}. 
\end{aligned}
\eeq
\e{prop}

In the next few subsections, we prove Proposition \ref{prop:keygrid}. Afterwards, in Subsections \ref{ssect:2Doptimize} and \ref{ssect:DDoptimize}, we discuss how to choose the parameters $c_i$ and $d_i$ to optimize the local gap threshold $\frac{K_0+K_3-2K_1}{K_4}$ which will complete the proof of Theorem \ref{thm:main}.

\subsection{Expansion of the auxiliary operator}\label{ssect:expansionaux}
For every site $t\in \Lambda_L$, we can define the translated box 
$$
B_t  = \{x \in \Lambda_L : x-t \in \mathcal{B}_\ell\}
$$
Recall the we use periodic boundary conditions on $\Lam_L$ but open boundary conditions on $B_\ell$.
We denote the corresponding unweighted and weighted Hamiltonians (the latter defined analogously to \eqref{eq:WBgrid}) by $H_{B_t}$ and $W_{B_t}$, respectively, which can be obtained by conjugating $H_{\mathcal{B}_\ell},W_{\mathcal{B}_\ell}$ with translation. Since translation is unitary, the spectral gap is independent of $t$. We consider the auxiliary operator
$
A=
\sum_{t \in \Lambda_L} W_{B_t}^2.
$
This operator can be written as a weighted sum of anti-commutators $\{h_{x,e_j}, h_{y,e_k}\}= h_{x,e_j}h_{y,{e_k}}+h_{y,e_k}h_{x,e_j}$, i.e. 
\begin{equation}\label{eq:wsum}
    A=\sum_{t\in \Lambda_L} W_{B_t}^2 =\sum_{(x, j) \in \Lambda_L \times [D]} w_{x,j} h_{x,e_j}^2 + \sum_{\substack{(x, j), (y, k) \in \Lambda_L \times [D]: \\(x,e_k) \neq (y,e_k)}} w_{x,y,j,k} \{h_{x,e_j}, h_{y,e_k}\}.
\end{equation}
where we denoted $[D]=\{1,\ldots,D\}$ for appropriate effective weights $w_{x,j}$ and $w_{x,y,j,k} $ to be investigated in the following. Computing these coefficient will lead us to the explicit formula \eqref{eq:Arewrite} below.

First, we have
$$
w_{x,j}=\left( \sum_{i=0}^{\ell-1} c_i^2\right) \left( \sum_{i=0}^\ell d_i^2\right)^{D-1}=K_0
$$
Combining this with the fact that $h_{x,e_j}^2 = h_{x,e_j}$ (since the $h_{x,e_j}$ are projectors) we find 
\begin{align*}
\sum_{(x, e_j) \in \Lambda_L \times [D]} w_{x,j} h_{x,e_j}^2 
= K_0 H_{\Lam_L},
\end{align*}
which is the first term in \eqref{eq:Arewrite}.

Next, we consider the off-diagonal weights $w_{x,y,j,k}$ in \eqref{eq:wsum}.  We write $x = (x_1, \dotsc, x_d), y = (y_1, \dotsc, y_d)$ and $z_j = |x_j-y_j|$. 

\textit{Case $j=k$.} This means that the two relevant edges are collinear. Then, 
\beq\label{eq:generalcollinear}
w_{x,y,j,j} = \left(\sum_{i=0}^{\ell-1-z_j} c_i c_{i+z_j} \right) \prod_{\substack{2\leq m\leq D:\\ m\neq j}} \left(\sum_{j=0}^{\ell-z_m} d_j d_{j+z_m} \right)
\eeq
Here, we use the assumption that $L>2\ell$, as otherwise one might have terms such as $d_id_{i+\ell-z_m}$ in the sums. We claim the expressions $\sum_{i=0}^{\ell-1-z_j} c_i c_{i+z_j} $ and $\sum_{j=0}^{\ell-z_m} d_j d_{j+z_m}$ are monotonically decreasing in the distances $z_1,\ldots,z_D$ under our assumptions on $c_i,d_i$. To see this, suppose we choose $a_i, b_i, 0 \leq i \leq n$, such that the $a_i = 1_{i \in [j,\ell-j]}, b_i = 1_{i \in [k, n-k]}$. Then $\sum_{i=0}^{\ell-z_m} a_ib_{i+z_m}$ is decreasing in $z_k$, and then linearity shows that $\sum_{i=0}^{\ell-z_k} d_i d_{i+z_k}$ is decreasing in $z_k$. An identical argument holds for the $c_i$. 

Since $x\neq y$, we have $\sum_{1\leq m\leq D} z_m\geq 1$ which is minimal for $\sum_{1\leq m\leq D} z_m\geq 1=1$. In other words,  the maximum value of $w_{x,y,j,j}$ when $x$ and $y$ are nearest neighbors. There are two possibilities for this: (i) both edges are incident to the same vertex, i.e. the edge pair is of the shape
\begin{tikzpicture}
\draw[thick] (0,.15) -- (.3,.15);
\draw[thick] (.3,.15) -- (.6,.15);
\tkzDefPoint(0,.15){A}
\tkzDefPoint(.3,.15){B}
\tkzDefPoint(.6,.15){C}
\foreach \n in {A,B,C}
  \node at (\n)[circle,fill,inner sep=1pt]{};
\end{tikzpicture}
, which translates to $z_j=1$ and all other $z_m=0$. (ii) Both edges are parallel but do not overlap at a vertex, i.e., $y = x+e_{m_0}$ for some $m_0 \neq j$. We  represent this case by the edge pair
\resizebox{.6cm}{!}{
\begin{tikzpicture}
\draw[thick] (0,0) -- (.3,0);
\draw[thick] (0,.3) -- (.3,.3);
\tkzDefPoint(0,0){A}
\tkzDefPoint(.3,0){B}
\tkzDefPoint(0,.3){C}
\tkzDefPoint(.3,.3){D}
\foreach \n in {A,B,C,D}
  \node at (\n)[circle,fill,inner sep=1pt]{};
\end{tikzpicture}
}. 
This case translates to $z_{m_0}=1$ for some $m_0\neq j$ and all other $z_m=0$. Distinguishing these cases in \eqref{eq:generalcollinear}, we conclude that
$$
w_{x,y,j,j}\leq \max\{K_1,K_2\},
$$
for the effective constants
$$
K_1= \left(\sum_{i=0}^{\ell-2}c_ic_{i+1}\right)\left(\sum_{i=0}^\ell d_i^2 \right)^{D-1},\quad 
K_2=\left(\sum_{i=0}^{\ell-1}c_i^2\right)\left(\sum_{i=0}^{\ell-1}d_id_{i+1} \right)\left(\sum_{i=0}^\ell d_i^2 \right)^{D-2}.
$$

\textit{Case $j \neq k$.} Without loss of generality, we assume $j=1, k=2$, in which case 
$$
w_{x,y,1,2} =  \left(\sum_{i=0}^{\ell-z_1-1} c_id_{i+z_1} \right)\left(\sum_{i=0}^{\ell-z_2-1} c_id_{i+z_2} \right)\prod_{k=3}^D \left(\sum_{i=0}^{\ell-z_k} d_i d_{i+z_k} \right).
$$

Via a similar argument as before, the above quantity is decreasing in $z_1, z_2$, and as a result the maximum value of $w_{x,y,1,2}$ is achieved by adjacent edges of the shape 
\resizebox{.6cm}{!}{
\begin{tikzpicture}
\draw[thick] (0,.3) -- (.3,.3);
\draw[thick] (.3,.3) -- (.3,0);
\tkzDefPoint(0,.3){A}
\tkzDefPoint(.3,.3){B}
\tkzDefPoint(.3,0){C}
\foreach \n in {A,B,C}
  \node at (\n)[circle,fill,inner sep=1pt]{};
  \end{tikzpicture}
  }.
  This yields the bound
\begin{equation}\label{eqn:k3}
w_{x,y,1,2}\leq K_3= \left( \sum_{i=0}^{\ell-1}c_id_i\right)^2 \left(\sum_{i=0}^{\ell} d_i^2 \right)^{D-2},\qquad j\neq k.
\end{equation}
The fact that the upper bound is $K_3$ for all choices of orientation of the edge shape 
\resizebox{.6cm}{!}{
\begin{tikzpicture}
\draw[thick] (0,.3) -- (.3,.3);
\draw[thick] (.3,.3) -- (.3,0);
\tkzDefPoint(0,.3){A}
\tkzDefPoint(.3,.3){B}
\tkzDefPoint(.3,0){C}
\foreach \n in {A,B,C}
  \node at (\n)[circle,fill,inner sep=1pt]{};
  \end{tikzpicture}
  }
  follows from the symmetry of the $c_i, d_i$ and relabelling the summation indices appropriately.

Thus, we can expand the second term in \eqref{eq:wsum} to obtain
\beq\label{eq:Arewrite}
A  =K_0H_{\Lam_L}+K_1 \sum_{\begin{tikzpicture}
\draw[thick] (0,.15) -- (.3,.15);
\draw[thick] (.3,.15) -- (.6,.15);
\tkzDefPoint(0,.15){A}
\tkzDefPoint(.3,.15){B}
\tkzDefPoint(.6,.15){C}
\foreach \n in {A,B,C}
  \node at (\n)[circle,fill,inner sep=1pt]{};
\end{tikzpicture}}  \{h_{x,e_j}, h_{y,e_k}\} 
+ K_2\sum_{\begin{tikzpicture}
\draw[thick] (-.3,-.1) -- (0,-.1);
\draw[thick] (-.3,-.4) -- (0,-.4);
\tkzDefPoint(0,-.1){A}
\tkzDefPoint(-.3,-.1){B}
\tkzDefPoint(0,-.4){C}
\tkzDefPoint(-.3,-.4){D}
\foreach \n in {A,B,C,D}
  \node at (\n)[circle,fill,inner sep=1pt]{};
\end{tikzpicture}}   \{h_{x,e_j}, h_{y,e_k}\}  +K_3
\sum_{\resizebox{.4cm}{!}{\begin{tikzpicture}
\draw[thick] (-.3,-.1) -- (0,-.1);
\draw[thick] (0,-.1) -- (0,-.4);
\tkzDefPoint(0,-.1){A}
\tkzDefPoint(-.3,-.1){B}
\tkzDefPoint(0,-.4){C}
\foreach \n in {A,B,C}
  \node at (\n)[circle,fill,inner sep=1pt]{};
\end{tikzpicture}}}   \{h_{x,e_j}, h_{y,e_k}\} + R_A,
\eeq
where the sums are over all edge pairs of the indicated type, and $R_A$ consists of the anti-commutators corresponding to all remaining edge pairs. The monotonicity established above and our assumption that $K_3 \geq \max\{K_1,K_2\}$ imply that the positive coefficient of any particular commutator $\{h_{x,e_j}, h_{y,e_k}\} $ in $R_G$ is bounded above by $K_3$. 

\subsection{Comparison to the squared Hamiltonian}
We may now similarly expand $H_{\Lam_L}^2$, writing 
\begin{equation}\label{eqn:hsum}
H_{\Lam_L}^2 = H_{\Lam_L} + \sum_{\begin{tikzpicture}
\draw[thick] (0,.15) -- (.3,.15);
\draw[thick] (.3,.15) -- (.6,.15);
\tkzDefPoint(0,.15){A}
\tkzDefPoint(.3,.15){B}
\tkzDefPoint(.6,.15){C}
\foreach \n in {A,B,C}
  \node at (\n)[circle,fill,inner sep=1pt]{};
\end{tikzpicture}} \{h_{x,e_j}, h_{y,e_k}\} + \sum_{\begin{tikzpicture}
\draw[thick] (-.3,-.1) -- (0,-.1);
\draw[thick] (-.3,-.4) -- (0,-.4);
\tkzDefPoint(0,-.1){A}
\tkzDefPoint(-.3,-.1){B}
\tkzDefPoint(0,-.4){C}
\tkzDefPoint(-.3,-.4){D}
\foreach \n in {A,B,C,D}
  \node at (\n)[circle,fill,inner sep=1pt]{};
\end{tikzpicture}} \{h_{x,e_j}, h_{y,e_k}\}+ \sum_{\begin{tikzpicture}
\draw[thick] (-.3,-.1) -- (0,-.1);
\draw[thick] (0,-.1) -- (0,-.4);
\tkzDefPoint(0,-.1){A}
\tkzDefPoint(-.3,-.1){B}
\tkzDefPoint(0,-.4){C}
\foreach \n in {A,B,C}
  \node at (\n)[circle,fill,inner sep=1pt]{};
\end{tikzpicture}} \{h_{x,e_j}, h_{y,e_k}\} + R_{\Lam_L}
\end{equation}
where again $R_{\Lam_L}$ contains the anti-commutators corresponding to all remaining edge pairs. Since the edge pairs entering $R_{\Lam_L}$ do not share a vertex, the corresponding projections commute and satisfy the positivity $\{h_{x,e_j}, h_{y,e_k}\}\geq 0$.  Considering that the coefficient of any particular commutator in $R_A$ is bounded above by $K_3$, we conclude that the operator $K_3 R_{\Lam_L} - R_A$
is a positive linear combination of positive operators. Analogous reasoning yields 
$$
(K_3-K_2)\sum_{\begin{tikzpicture}
\draw[thick] (-.3,-.1) -- (0,-.1);
\draw[thick] (-.3,-.4) -- (0,-.4);
\tkzDefPoint(0,-.1){A}
\tkzDefPoint(-.3,-.1){B}
\tkzDefPoint(0,-.4){C}
\tkzDefPoint(-.3,-.4){D}
\foreach \n in {A,B,C,D}
  \node at (\n)[circle,fill,inner sep=1pt]{};
\end{tikzpicture}} \{h_{x,e_j}, h_{y,e_k}\}\geq 0.$$

Combining these estimates with our expansions \eqref{eq:Arewrite} and \eqref{eqn:hsum}, we obtain the operator inequality 
\begin{align*}
K_3H^2_{\Lam_L} - \sum_{l \in \Lambda_L} W_{B_l}^2 &\geq (K_3-K_0) H_{\Lam_L}+(K_3-K_1) \sum_{
\begin{tikzpicture}
\draw[thick] (0,.15) -- (.3,.15);
\draw[thick] (.3,.15) -- (.6,.15);
\tkzDefPoint(0,.15){A}
\tkzDefPoint(.3,.15){B}
\tkzDefPoint(.6,.15){C}
\foreach \n in {A,B,C}
  \node at (\n)[circle,fill,inner sep=1pt]{};
\end{tikzpicture}} \{h_{x,e_j}, h_{y,e_k}\}.
\end{align*}

To treat the last term, we draw on an idea of \cite{L2} and apply the operator Cauchy-Schwarz inequality in the form 
$$
-\{h_{x,e_j}, h_{y,e_k}\} \leq (-h_{x,e_j})^2 + h_{y,e_k}^2 = h_{x,e_j}  + h_{y,e_k}
$$
which implies
$$
\sum_{\begin{tikzpicture}
\draw[thick] (0,.15) -- (.3,.15);
\draw[thick] (.3,.15) -- (.6,.15);
\tkzDefPoint(0,.15){A}
\tkzDefPoint(.3,.15){B}
\tkzDefPoint(.6,.15){C}
\foreach \n in {A,B,C}
  \node at (\n)[circle,fill,inner sep=1pt]{};
\end{tikzpicture}} \{h_{x,e_j}, h_{y,e_k}\} \geq - 2H_{\Lam_L},
$$
because each edge is contained in precisely two bond pairs of the shape \begin{tikzpicture}
\draw[thick] (0,.15) -- (.3,.15);
\draw[thick] (.3,.15) -- (.6,.15);
\tkzDefPoint(0,.15){A}
\tkzDefPoint(.3,.15){B}
\tkzDefPoint(.6,.15){C}
\foreach \n in {A,B,C}
  \node at (\n)[circle,fill,inner sep=1pt]{};
\end{tikzpicture}.
We have shown that
\begin{equation} \label{eqn:fundamentalsum}
K_3H^2_{\Lam_L} - \sum_{t \in \Lambda_L} W_{B_t}^2 \geq (2K_1-K_0-K_3) H_{\Lam_L}
\end{equation}
which we point out already bears some similarity with \eqref{eq:propclaim}.

\subsection{Translation-invariant excited states}
Finally, we compare $\sum_{t \in \Lambda_L} W^2_{B_t}$ with $H_{\Lam_L}$. The naive estimate involving the spectral gap of the unweighted operator times the smallest coefficient $\min\{c_0,d_0\}$ is insufficient. Instead, we use the following refinement which leverages translation-invariance and is directly inspired by \cite[Lemma 4]{GM}.

\begin{lm}\label{lm:translationgrid}

There exists a normalized state $|\phi\rangle$ satisfying $H_{\Lam_L}|\phi\rangle=\gamma_{\Lam_L}^{\mathrm{per}}|\phi\rangle$ such that
\begin{equation} \label{eqn:averaging}
\langle \phi , W^2_{B_t}\phi \rangle  \geq \gamma_\ell \left(\min_{1\leq j\leq D}  \frac{1}{\ell(\ell+1)^{D-1}}\sum_{\substack{x\in B_t \\ 0 \leq x_j \leq \ell-1}} w_{x,e_j} \right)\langle \phi , W_{B_t}\phi \rangle,\qquad t\in\Lam_L.
\end{equation}
\end{lm}
\begin{proof}[Proof of Lemma \ref{lm:translationgrid}]
Let $T_i$ be spatial translation in the direction $i$. Then, since the $\{T_i\}_{1\leq i\leq D}$ and $H_{\Lam_L}$ all mutually commute, we can simultaneously diagonalize these operators. Thus, we can find a state $|\phi\rangle$ lying in the $\gamma_L^{\mathrm{per}}$-eigenspace of $H_{\Lam_L}$ such that $|\phi \rangle$ is also an eigenvector of all the $T_i$'s. For any two bonds $(x,e_j), (y,e_j)$, by conjugating with the translation operators we find that the energy 
\beq\label{eq:energyid}
\langle \phi | h_{x,e_j}   |\phi \rangle  = \langle \phi | h_{y,e_j}   |\phi \rangle
\eeq
is constant across all bonds. 

Next, let $P$ be the projector onto the ground state of $W_{B_t}$, and write $P^\perp$ for $I-P$. If $P^\perp |\phi \rangle =0$ then both sides in \eqref{eqn:averaging} vanish by frustration-freeness and the claim is trivial. Hence, we may assume $P^\perp |\phi \rangle \neq 0$ and write $|\hat{\phi}\rangle$ for the corresponding normalized state which satisfies $\langle \hat{\phi} , X \hat{\phi} \rangle = \langle \phi , X \phi \rangle\langle \phi , P^\perp  \phi \rangle$ for any bounded operator $X$.

Consider the case $t=0$, i.e., $B_t=\mathcal B_\ell$. From Jensen's inequality
\begin{align*}
    \langle \phi , W_{\mathcal{B}_\ell}^2 \phi \rangle &= \langle \hat{\phi}, W_{\mathcal{B}_\ell}^2 \hat{\phi} \rangle \langle \phi , P^\perp \phi \rangle 
    \geq (\langle \hat{\phi} , W_{\mathcal{B}_\ell} \hat{\phi}\rangle )^2 \langle \phi , P^\perp  \phi \rangle 
    = \langle  \hat{\phi} , W_{\mathcal{B}_\ell}  \hat{\phi} \rangle\langle \phi,W_{\mathcal{B}_\ell}\phi \rangle.
\end{align*}

The energy identity \eqref{eq:energyid} extends to $|\hat{\phi}\rangle$, i.e., $
\langle \hat{\phi} | h_{x,e_j}   |\hat{\phi} \rangle  = \langle \hat{\phi} | h_{y,e_j}   |\hat{\phi}\rangle$
which lets us average over the weights in $W_{\mathcal{B}_\ell}$, i.e., 
\begin{align*}
 \langle  \hat{\phi} , W_{\mathcal{B}_\ell}  \hat{\phi} \rangle &=  \l\langle  \hat{\phi} , \sum_{1 \leq j \leq D}\sum_{\substack{x\in \mathcal{B}_\ell: \\ 0 \leq x_j \leq \ell-1}}  w_{x,e_j} h_{x,e_j} \hat{\phi} \r\rangle \\
 &= \l\langle \hat{\phi} , \sum_{1 \leq j \leq D}\sum_{\substack{x\in \mathcal{B}_\ell: \\ 0 \leq x_j \leq \ell-1}}  \left(\frac{1}{\ell(\ell+1)^{D-1}}\sum_{\substack{y\in \mathcal{B}_\ell: \\ 0 \leq y_j \leq \ell-1}} w_{y,e_j}\right) h_{x,e_j}  \hat{\phi} \r\rangle \\
 &\geq  \left( \min_{1\leq j\leq D} \frac{1}{\ell(\ell+1)^{D-1}} \sum_{\substack{y\in \mathcal{B}_\ell: \\ 0 \leq y_j \leq \ell-1}}  w_{y,e_j} \right) \langle \hat{\phi}| H_{\mathcal{B}_\ell} |\hat{\phi} \rangle \\
 &\geq \gamma_\ell \left( \min_{1\leq j\leq D}\frac{1}{\ell(\ell+1)^{D-1}} \sum_{\substack{y\in \mathcal{B}_\ell: \\ 0 \leq y_j \leq \ell-1}}  w_{y,e_j} \right) 
\end{align*}
where in the last step we used the fact that $\hat{\phi}$ has been chosen so that it is orthogonal to the ground state of $W_{\mathcal{B}_\ell}$ and thus also of $H_{\mathcal{B}_\ell}$. (The weights do not change the zero eigenspace since they are positive.) The argument for $t\neq 0$ is analogous. This proves Lemma \ref{lm:translationgrid}.
\end{proof}

\subsection{Proof of Proposition \ref{prop:keygrid}}
We take the expectation value of \eqref{eqn:fundamentalsum} in the normalized state $|\phi\rangle$ whose existence is guaranteed by Lemma \ref{lm:translationgrid}. Using $H_{\Lam_L}|\phi\rangle=\gamma_{\Lam_L}^{\mathrm{per}}|\phi\rangle$, we obtain
\beq\label{eq:appliedlm}
K_3(\gamma_{\Lam_L}^{\mathrm{per}})^2 -\l\langle\phi, \sum_{t \in \Lambda_L} W_{B_t}^2\phi\r\rangle
 \geq (2K_1-K_0-K_3) \gamma_{\Lam_L}^{\mathrm{per}}.
\eeq

For our choice of weights, the average $\frac{1}{\ell(\ell+1)^{D-1}}\sum_y w_{y,e_j}$ in \eqref{eqn:averaging} is independent of direction $j$. Expressing this average in terms of the coefficients $c_i,d_i$ gives 
\begin{equation} \label{eqn:avgweights}
 \langle  \phi | W_{\mathcal{B}_\ell}^2 |  \phi \rangle  \geq \frac{\gamma_\ell}{\ell(\ell+1)^{D-1}} \left(\sum_{i=0}^{\ell-1}c_i \right)\left( \sum_{i=0}^\ell d_i\right)^{D-1} \langle   \phi | W_{\mathcal{B}_\ell} |  \phi \rangle
\end{equation}
The same inequality holds for all translates $W^2_{B_t}$ thanks to translation-invariance. Summing \eqref{eqn:avgweights} over all translates $W^2_{B_t}$, we obtain
\begin{align*}
     \sum_{t \in \Lambda_L} \langle  \phi | W_{B_t}^2|  \phi \rangle &\geq \frac{\gamma_\ell}{\ell(\ell+1)^{D-1}} \left(\sum_{i=0}^{\ell-1}c_i \right)\left( \sum_{i=0}^\ell d_i\right)^{D-1} \l\langle   \phi , \sum_{t \in \Lambda_L} W_{B_t}   \phi \r\rangle \\
     &= \frac{\gamma_\ell}{\ell(\ell+1)^{D-1}} \left(\sum_{i=0}^{\ell-1}c_i \right)^2\left( \sum_{i=0}^\ell d_i\right)^{2(D-1)} \langle \phi ,H_{\Lam_L},  \phi\rangle\\
     &=\gam_\ell \gam_L^{\mathrm{per}} K_4, 
\end{align*}
where the second step uses that in the sum $\sum_{t \in \Lambda_L} W_{B_t}$, each projector $h_{x,e_j}$ appears with coefficient $\sum_x w_{x,e_j}=K_0$ and the last step uses that $H_{\Lam_L}|\phi\rangle=\gamma_{\Lam_L}^{\mathrm{per}}|\phi\rangle$. 

Returning to \eqref{eq:appliedlm} und using that $\gam_L^{\mathrm{per}}>0$ by definition, we have shown that
\begin{equation}\label{eqn:rawthm1}
    \gamma_L^{\mathrm{per}} \geq \frac{K_4}{K_3}\left(\gamma_\ell - \frac{K_0+K_3-2K_1  }{K_4} \right)
\end{equation}
which proves Proposition \ref{prop:keygrid}.
\qed

\subsection{Parameters in 2 dimensions}\label{ssect:2Doptimize}
In this section we consider $D = 2$. 

\be{proof}[Proof of Theorem \ref{thm:main} for $D=2$]
We have Proposition \ref{prop:keygrid}  for all weights $c_i,d_i$ satisfying requirements (i)-(iii) above. For $D=2$, 
\begin{align*}
K_0 &= \left( \sum_{i=0}^{\ell-1} c_i^2\right) \left( \sum_{i=0}^\ell d_i^2\right),\qquad \quad\;\;
K_1 = \left(\sum_{i=0}^{\ell-2}c_ic_{i+1}\right)\left(\sum_{i=0}^\ell d_i^2 \right), \\
K_2 &= \left(\sum_{i=0}^{\ell-1}c_i^2\right)\left(\sum_{i=0}^{\ell-1}d_id_{i+1} \right),\qquad 
K_3 = \left( \sum_{i=0}^{\ell-1}c_id_i\right)^2,  \\
K_4 &= \frac{1}{\ell^2+\ell} \left(\sum_{i=0}^{\ell-1}c_i \right)^2\left( \sum_{i=0}^\ell d_i\right)^{2} 
\end{align*}
In view of \eqref{eq:propclaim}, the main goal is to choose the coefficients $c_i,d_i$ such that the threshold $$ \frac{K_0+K_3-2K_1  }{K_4}$$ is as small as possible. 

We set
\beq\label{eq:cjchoice}
\begin{aligned}
c_j &= \ell+(\ell -1)j-j^2,\qquad&\textnormal{for } j=0,\ldots,\ell-1, \\
d_j &= (1- \lambda )(\ell+1+\ell j-j^2) +  \lambda \frac{(\ell+2)^2}{4},\qquad&\textnormal{for } j=0,\ldots,\ell
\end{aligned}
\eeq
with interpolation parameter $ \lambda =\frac{2(\sqrt{2}-1)}{\ell}$. These satisfy the requirements (i)-(iii) from above.


It remains to explicit calculate the threshold $ \frac{K_0+K_3-2K_1  }{K_4}$ in terms of the polynomials $K_0,K_1,\ldots,K_4$ and to use various elementary error estimates. 

\textit{Asymptotics.} Before we present the rigorous error estimates, let us check that the result is correct to leading order in $\ell$. First, to check the constraint $K_3 \geq K_1, K_2$, we consider
\beq\label{eq:poscond}
\begin{aligned}
    K_3-K_1 &= \frac{2+\sqrt{2}}{360}\ell^7 +  \mathcal O(\ell^{6})\\
    K_3-K_2 &= \frac{5-3\sqrt{2}}{360}\ell^7 +  \mathcal O(\ell^{6}),
\end{aligned}
\eeq
which holds asymptotically as $\ell\to\infty$. We see that for large enough $\ell$ it will indeed be the case that $K_3 \geq K_1, K_2$. The threshold can be written as
\begin{equation}
\frac{K_0+K_3-2K_1}{K_4} = \frac{(\ell^2+\ell)(K_0+K_3-2K_1)}{\left(\sum_{i=0}^{\ell-1}c_i \right)^2\left( \sum_{i=0}^\ell d_i\right)^{2} }
\end{equation}
Here, after cancellations, the numerator is given by a degree 10 polynomial with leading coefficients of the form
\beq\label{eq:asymp1}
(\ell^2+\ell)(K_0+K_3-2K_1) = \frac{1}{180}\ell^{10} + \frac{28+3\sqrt{2}}{360}\ell^9 + \mathcal O(\ell^{10}).
\eeq
The denominator is a degree 12 polynomial with leading coefficients of the form
\beq\label{eq:asymp2}
\left(\sum_{i=0}^{\ell-1}c_i \right)^2\left( \sum_{i=0}^\ell d_i\right)^{2} = \frac{1}{1296}\ell^{12} + \frac{8+\sqrt{2}}{648}\ell^{11} +  \mathcal O(\ell^{10}).
\eeq
The asymptotic statements \eqref{eq:asymp1} and \eqref{eq:asymp2} already establish that the leading coefficient  is as displayed in Theorem \ref{thm:main}, i.e.,
\beq\label{eq:asympconclusion}
 \frac{K_0+K_3-2K_1}{K_4}<   \frac{36}{5\ell^2} + \mathcal O\left(\frac{1}{\ell^3}\right) 
 \eeq
 
 \textit{Error estimates.} We now present explicit error estimates that yield the claimed explicit error term in \eqref{eq:asympconclusion}
First, by explicitly computing all coefficients in \eqref{eq:poscond}, one can establish $K_3 \geq K_1, K_2$ as soon as $\ell \geq 3$.

Furthermore, through various elementary estimates we obtain the following precise versions of \eqref{eq:asymp1} and \eqref{eq:asymp2},
$$
(\ell^2+\ell)(K_0+K_1-2K_3)  < \frac{1}{180}\ell^{10} + \frac{1}{5}\ell^9,\qquad \ell\geq 10,
$$
and
$$
\left(\sum_{i=0}^{\ell-1}c_i \right)^2\left( \sum_{i=0}^\ell d_i\right)^{2}  > \frac{1}{1296}\ell^{12},\qquad \ell\geq 10.
$$
This allows to estimate the threshold by
$$
    \frac{K_0+K_3-2K_1}{K_4}< \frac{36}{5\ell^2} + \frac{1296}{5\ell^3}, \qquad\ell \geq 10.
$$
We recall Proposition \ref{prop:keygrid} and use 
$$
\frac{K_4}{K_3}  \geq \frac{25}{36}, \qquad\ell \geq 10.
$$
to obtain
\begin{equation}
\gamma_L^{\mathrm{per}} \geq \frac{25}{36}\left(\gamma_\ell -   \frac{36}{5\ell^2} - \frac{1296}{5\ell^3} \right), \qquad\ell \geq 10,
\end{equation}
which proves Theorem \ref{thm:main} for $D=2$.
\e{proof}

Let us briefly discuss the heuristics underlying the choice of coefficients \eqref{eq:cjchoice}, keeping in mind the goal of minimizing the threshold $\frac{K_0+K_3-2K_1}{K_4} $. As a first benchmark, notice that if we choose the coefficients constant, we obtain the threshold scaling $\sim \frac{1}{\ell}$. Considering the numerator $K_0+K_3-2K_1$, we complete a square to write $K_0-2K_1$ as a sum over $(c_{i+1}-c_i)^2$ which by taking a variational derivative naturally leads to $c_i$'s being an eigenfunction of the discrete 1D Laplacian, i.e., quadratic polynomials, cf.\ \eqref{eq:cjchoice}. Conveniently, this is a class of functions depending on only $3$ parameters over which we can then optimize. The other term in the numerator, $K_3= \left( \sum_{i=0}^{\ell-1}c_id_i\right)^2 $ can be viewed as a scalar product between the vector of $d_i$'s and the vector of $c_i$'s, so the Cauchy-Schwarz inequality says this is maximal if these vectors are collinear.
Finally, regarding the choice of interpolation parameter $ \lambda =\frac{2(\sqrt{2}-1)}{\ell}$, one can compute first that for $ \lambda = \frac{C}{\ell}$ each term $K_0, K_1, K_3$ is a degree 10 polynomials in $\ell$, with leading behavior $\frac{1}{900}\ell^8 + \frac{31}{1800}\ell^7$ independent of $C$. Thus $K_0+K_3-2K_1$ is a degree 8 polynomial in $\ell$. Since $K_4$ is $\frac{1}{\ell^2+\ell}$ times a degree 12 polynomial in $\ell$, this method indeed yields an $O\left(\frac{1}{\ell^2}\right)$ scaling and it remains to choose $C$ which we do by electing to minimize the leading term of the degree 8 polynomial $K_0+K_3-2K_1$ in the numerator subject to the inequalities $K_3\geq K_1, K_2$. 

For small $\ell$, namely $\ell\leq 9$, the leading term is no longer dominant and other choices of $C$ (respectively $t$) yield appreciably smaller threshold. In this case, we instead rely on elementary optimization which yields the numbers displayed in Table \ref{tab:table1}. 


 %
%

\subsection{Parameters in $D$ dimensions}\label{ssect:DDoptimize}

\be{proof}[Proof of Theorem \ref{thm:main}]
Let $D\geq 2$. We apply Proposition \ref{prop:keygrid} and analyze the formulae \eqref{eq:Kdefns}.
Our goal is again to minimize $\frac{K_0+K_3-2K_1}{K_4}$. With the exception of $K_4$, all other expressions given above are equal to their 2D counterparts scaled by an additional factor of $\left( \sum_{i=0}^\ell d_i^2\right)^{D-2}$. $K_4$ is equal to its 2D counterpart scaled by a factor of $(\ell+1)^{-(D-2)}\left( \sum_{i=0}^\ell d_i\right)^{2(D-2)}$. Thus, if we choose the coefficients $c_i,d_i$ as in \eqref{eq:cjchoice}, we obtain from \eqref{eq:asympconclusion} that
$$
\frac{K_0+K_3-2K_1}{K_4} =  \left(\frac{(\ell+1)\sum_{i=0}^\ell d_i^2 }{\left( \sum_{i=0}^\ell d_i\right)^2}\right)^{D-2}  \left( \frac{36}{5\ell^2}+O\left( \frac{1}{\ell^3}\right)\right).
$$
In order to derive error estimates, we note that both $(\ell+1)\sum_{i=0}^\ell d_i^2$ and $\left( \sum_{i=0}^\ell d_i\right)^2$ are degree 6 polynomials in $\ell$. By writing out these polynomials and using elementary estimates, we obtain the error estimate
$$
\frac{(\ell+1)\sum_{i=0}^\ell d_i^2 }{\left( \sum_{i=0}^\ell d_i\right)^2} \leq \frac{6}{5},\qquad \ell\geq 10.
$$
Similarly, one can show $\frac{K_4}{K_3} \geq \left( \frac{5}{6}\right)^D$, and so substituting in \eqref{eqn:rawthm1} we have
\begin{align*}
    \gamma_L^{\mathrm{per}} 
    &\geq \left(\frac{5}{6}\right)^D \left( \gamma_\ell -  \left(\frac{6}{5}\right)^D\left(\frac{5}{\ell^2}+\frac{300}{\ell^3} \right) \right),
\end{align*}
which proves Theorem 1. 
\e{proof}


\section{Proof of Theorem \ref{thm:mainhoney} on the honeycomb lattice}
The proof follows the same general line of argumentation as the proof of Theorem \ref{thm:main} on the Euclidean lattice. The main difference is the definition of the weighting scheme which takes into account the local lattice geometry as described next.

\subsection{Construction of weighted subsystem Hamiltonians}
We let $c_0, \dotsc c_{\ell-1}$ and $d_0, \dotsc d_{\ell}$ be coefficients satisfying Requirements (i)-(iii) in Subsection \ref{ssect:coefficients} whose exact values will be determined later. We then define $W_{\mathcal B_\ell}$ analogously to \eqref{eq:WBgrid} as the pointwise product of two simpler weighting schemes, now on the slanted grid $B_\ell$, cf.\ Figure \ref{fig:hexslant}. To avoid introducing excessive notation for implementing this formally on the honeycomb lattice, we summarize the honeycomb weighting scheme through Figure \ref{fig:WB3honey} which generalizes to arbitrary $\ell$ in an obvious way.

\vspace{.5cm}
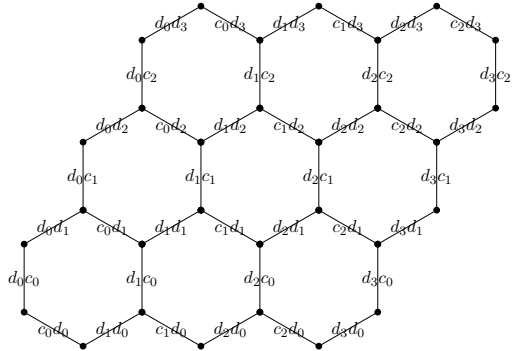
\begin{figure}[h]
\begin{center}
\resizebox{7cm}{!}{
\begin{tikzpicture}
\node at (1.2990,-0.7500) [circle,fill,inner sep=1.5pt]{};
\node at (1.2990,0.7500) [circle,fill,inner sep=1.5pt]{};
\node at (0.0000,1.5000) [circle,fill,inner sep=1.5pt]{};
\node at (-1.2990,0.7500) [circle,fill,inner sep=1.5pt]{};
\node at (-1.2990,-0.7500) [circle,fill,inner sep=1.5pt]{};
\node at (-0.0000,-1.5000) [circle,fill,inner sep=1.5pt]{};
\node at (2.5981,1.5000) [circle,fill,inner sep=1.5pt]{};
\node at (2.5981,3.0000) [circle,fill,inner sep=1.5pt]{};
\node at (1.2990,3.7500) [circle,fill,inner sep=1.5pt]{};
\node at (-0.0000,3.0000) [circle,fill,inner sep=1.5pt]{};
\node at (0.0000,1.5000) [circle,fill,inner sep=1.5pt]{};
\node at (1.2990,0.7500) [circle,fill,inner sep=1.5pt]{};
\node at (3.8971,3.7500) [circle,fill,inner sep=1.5pt]{};
\node at (3.8971,5.2500) [circle,fill,inner sep=1.5pt]{};
\node at (2.5981,6.0000) [circle,fill,inner sep=1.5pt]{};
\node at (1.2990,5.2500) [circle,fill,inner sep=1.5pt]{};
\node at (1.2990,3.7500) [circle,fill,inner sep=1.5pt]{};
\node at (2.5981,3.0000) [circle,fill,inner sep=1.5pt]{};
\node at (3.8971,-0.7500) [circle,fill,inner sep=1.5pt]{};
\node at (3.8971,0.7500) [circle,fill,inner sep=1.5pt]{};
\node at (2.5981,1.5000) [circle,fill,inner sep=1.5pt]{};
\node at (1.2990,0.7500) [circle,fill,inner sep=1.5pt]{};
\node at (1.2990,-0.7500) [circle,fill,inner sep=1.5pt]{};
\node at (2.5981,-1.5000) [circle,fill,inner sep=1.5pt]{};
\node at (5.1962,1.5000) [circle,fill,inner sep=1.5pt]{};
\node at (5.1962,3.0000) [circle,fill,inner sep=1.5pt]{};
\node at (3.8971,3.7500) [circle,fill,inner sep=1.5pt]{};
\node at (2.5981,3.0000) [circle,fill,inner sep=1.5pt]{};
\node at (2.5981,1.5000) [circle,fill,inner sep=1.5pt]{};
\node at (3.8971,0.7500) [circle,fill,inner sep=1.5pt]{};
\node at (6.4952,3.7500) [circle,fill,inner sep=1.5pt]{};
\node at (6.4952,5.2500) [circle,fill,inner sep=1.5pt]{};
\node at (5.1962,6.0000) [circle,fill,inner sep=1.5pt]{};
\node at (3.8971,5.2500) [circle,fill,inner sep=1.5pt]{};
\node at (3.8971,3.7500) [circle,fill,inner sep=1.5pt]{};
\node at (5.1962,3.0000) [circle,fill,inner sep=1.5pt]{};
\node at (6.4952,-0.7500) [circle,fill,inner sep=1.5pt]{};
\node at (6.4952,0.7500) [circle,fill,inner sep=1.5pt]{};
\node at (5.1962,1.5000) [circle,fill,inner sep=1.5pt]{};
\node at (3.8971,0.7500) [circle,fill,inner sep=1.5pt]{};
\node at (3.8971,-0.7500) [circle,fill,inner sep=1.5pt]{};
\node at (5.1962,-1.5000) [circle,fill,inner sep=1.5pt]{};
\node at (7.7942,1.5000) [circle,fill,inner sep=1.5pt]{};
\node at (7.7942,3.0000) [circle,fill,inner sep=1.5pt]{};
\node at (6.4952,3.7500) [circle,fill,inner sep=1.5pt]{};
\node at (5.1962,3.0000) [circle,fill,inner sep=1.5pt]{};
\node at (5.1962,1.5000) [circle,fill,inner sep=1.5pt]{};
\node at (6.4952,0.7500) [circle,fill,inner sep=1.5pt]{};
\node at (9.0933,3.7500) [circle,fill,inner sep=1.5pt]{};
\node at (9.0933,5.2500) [circle,fill,inner sep=1.5pt]{};
\node at (7.7942,6.0000) [circle,fill,inner sep=1.5pt]{};
\node at (6.4952,5.2500) [circle,fill,inner sep=1.5pt]{};
\node at (6.4952,3.7500) [circle,fill,inner sep=1.5pt]{};
\node at (7.7942,3.0000) [circle,fill,inner sep=1.5pt]{};
\draw(-1.2990, -0.7500)--(-0.0000, -1.5000) node[pos=.5]{$c_0d_0$};
\draw(0.0000, 1.5000)--(1.2990, 0.7500) node[pos=.5]{$c_0d_1$};
\draw(1.2990, 3.7500)--(2.5981, 3.0000) node[pos=.5]{$c_0d_2$};
\draw(2.5981, 6.0000)--(3.8971, 5.2500) node[pos=.5]{$c_0d_3$};
\draw(1.2990, -0.7500)--(2.5981, -1.5000) node[pos=.5]{$c_1d_0$};
\draw(2.5981, 1.5000)--(3.8971, 0.7500) node[pos=.5]{$c_1d_1$};
\draw(3.8971, 3.7500)--(5.1962, 3.0000) node[pos=.5]{$c_1d_2$};
\draw(5.1962, 6.0000)--(6.4952, 5.2500) node[pos=.5]{$c_1d_3$};
\draw(3.8971, -0.7500)--(5.1962, -1.5000) node[pos=.5]{$c_2d_0$};
\draw(5.1962, 1.5000)--(6.4952, 0.7500) node[pos=.5]{$c_2d_1$};
\draw(6.4952, 3.7500)--(7.7942, 3.0000) node[pos=.5]{$c_2d_2$};
\draw(7.7942, 6.0000)--(9.0933, 5.2500) node[pos=.5]{$c_2d_3$};
\draw(-1.2990, -0.7500)--(-1.2990, 0.7500) node[pos=.5]{$d_0c_0$};
\draw(0.0000, 1.5000)--(-0.0000, 3.0000) node[pos=.5]{$d_0c_1$};
\draw(1.2990, 3.7500)--(1.2990, 5.2500) node[pos=.5]{$d_0c_2$};
\draw(1.2990, -0.7500)--(1.2990, 0.7500) node[pos=.5]{$d_1c_0$};
\draw(2.5981, 1.5000)--(2.5981, 3.0000) node[pos=.5]{$d_1c_1$};
\draw(3.8971, 3.7500)--(3.8971, 5.2500) node[pos=.5]{$d_1c_2$};
\draw(3.8971, -0.7500)--(3.8971, 0.7500) node[pos=.5]{$d_2c_0$};
\draw(5.1962, 1.5000)--(5.1962, 3.0000) node[pos=.5]{$d_2c_1$};
\draw(6.4952, 3.7500)--(6.4952, 5.2500) node[pos=.5]{$d_2c_2$};
\draw(6.4952, -0.7500)--(6.4952, 0.7500) node[pos=.5]{$d_3c_0$};
\draw(7.7942, 1.5000)--(7.7942, 3.0000) node[pos=.5]{$d_3c_1$};
\draw(9.0933, 3.7500)--(9.0933, 5.2500) node[pos=.5]{$d_3c_2$};
\draw(-1.2990, 0.7500)--(-0.0000, 1.5000) node[pos=.5]{$d_0d_1$};
\draw(-0.0000, 3.0000)--(1.2990, 3.7500) node[pos=.5]{$d_0d_2$};
\draw(1.2990, 5.2500)--(2.5981, 6.0000) node[pos=.5]{$d_0d_3$};
\draw(-0.0000, -1.5000)--(1.2990, -0.7500) node[pos=.5]{$d_1d_0$};
\draw(1.2990, 0.7500)--(2.5981, 1.5000) node[pos=.5]{$d_1d_1$};
\draw(2.5981, 3.0000)--(3.8971, 3.7500) node[pos=.5]{$d_1d_2$};
\draw(3.8971, 5.2500)--(5.1962, 6.0000) node[pos=.5]{$d_1d_3$};
\draw(2.5981, -1.5000)--(3.8971, -0.7500) node[pos=.5]{$d_2d_0$};
\draw(3.8971, 0.7500)--(5.1962, 1.5000) node[pos=.5]{$d_2d_1$};
\draw(5.1962, 3.0000)--(6.4952, 3.7500) node[pos=.5]{$d_2d_2$};
\draw(6.4952, 5.2500)--(7.7942, 6.0000) node[pos=.5]{$d_2d_3$};
\draw(5.1962, -1.5000)--(6.4952, -0.7500) node[pos=.5]{$d_3d_0$};
\draw(6.4952, 0.7500)--(7.7942, 1.5000) node[pos=.5]{$d_3d_1$};
\draw(7.7942, 3.0000)--(9.0933, 3.7500) node[pos=.5]{$d_3d_2$};
\end{tikzpicture}
}
\caption{The weighted subsystem Hamiltonian $W_{\mathcal{B}_3}$ on slanted subgrids of the honeycomb lattice. Edges are labeled by their weights. This is the honeycomb analog of Figure \ref{fig:WB3grid}.}
\label{fig:WB3honey}
\end{center}
\end{figure}
\vspace{-.5cm}

As in the case of the square grid, the honeycomb weighting scheme displayed in Figure \ref{fig:WB3honey} is the pointwise product of the following two simpler weighting schemes. 

\vspace{.5cm}
\begin{figure}[h]
\begin{center}
\resizebox{11cm}{!}{
    \begin{minipage}{.4\textwidth}
    \begin{tikzpicture}
\node at (0.8660,-0.5000) [circle,fill,inner sep=1.5pt]{};
\node at (0.8660,0.5000) [circle,fill,inner sep=1.5pt]{};
\node at (0.0000,1.0000) [circle,fill,inner sep=1.5pt]{};
\node at (-0.8660,0.5000) [circle,fill,inner sep=1.5pt]{};
\node at (-0.8660,-0.5000) [circle,fill,inner sep=1.5pt]{};
\node at (-0.0000,-1.0000) [circle,fill,inner sep=1.5pt]{};
\node at (1.7321,1.0000) [circle,fill,inner sep=1.5pt]{};
\node at (1.7321,2.0000) [circle,fill,inner sep=1.5pt]{};
\node at (0.8660,2.5000) [circle,fill,inner sep=1.5pt]{};
\node at (-0.0000,2.0000) [circle,fill,inner sep=1.5pt]{};
\node at (0.0000,1.0000) [circle,fill,inner sep=1.5pt]{};
\node at (0.8660,0.5000) [circle,fill,inner sep=1.5pt]{};
\node at (2.5981,2.5000) [circle,fill,inner sep=1.5pt]{};
\node at (2.5981,3.5000) [circle,fill,inner sep=1.5pt]{};
\node at (1.7321,4.0000) [circle,fill,inner sep=1.5pt]{};
\node at (0.8660,3.5000) [circle,fill,inner sep=1.5pt]{};
\node at (0.8660,2.5000) [circle,fill,inner sep=1.5pt]{};
\node at (1.7321,2.0000) [circle,fill,inner sep=1.5pt]{};
\node at (2.5981,-0.5000) [circle,fill,inner sep=1.5pt]{};
\node at (2.5981,0.5000) [circle,fill,inner sep=1.5pt]{};
\node at (1.7321,1.0000) [circle,fill,inner sep=1.5pt]{};
\node at (0.8660,0.5000) [circle,fill,inner sep=1.5pt]{};
\node at (0.8660,-0.5000) [circle,fill,inner sep=1.5pt]{};
\node at (1.7321,-1.0000) [circle,fill,inner sep=1.5pt]{};
\node at (3.4641,1.0000) [circle,fill,inner sep=1.5pt]{};
\node at (3.4641,2.0000) [circle,fill,inner sep=1.5pt]{};
\node at (2.5981,2.5000) [circle,fill,inner sep=1.5pt]{};
\node at (1.7321,2.0000) [circle,fill,inner sep=1.5pt]{};
\node at (1.7321,1.0000) [circle,fill,inner sep=1.5pt]{};
\node at (2.5981,0.5000) [circle,fill,inner sep=1.5pt]{};
\node at (4.3301,2.5000) [circle,fill,inner sep=1.5pt]{};
\node at (4.3301,3.5000) [circle,fill,inner sep=1.5pt]{};
\node at (3.4641,4.0000) [circle,fill,inner sep=1.5pt]{};
\node at (2.5981,3.5000) [circle,fill,inner sep=1.5pt]{};
\node at (2.5981,2.5000) [circle,fill,inner sep=1.5pt]{};
\node at (3.4641,2.0000) [circle,fill,inner sep=1.5pt]{};
\node at (4.3301,-0.5000) [circle,fill,inner sep=1.5pt]{};
\node at (4.3301,0.5000) [circle,fill,inner sep=1.5pt]{};
\node at (3.4641,1.0000) [circle,fill,inner sep=1.5pt]{};
\node at (2.5981,0.5000) [circle,fill,inner sep=1.5pt]{};
\node at (2.5981,-0.5000) [circle,fill,inner sep=1.5pt]{};
\node at (3.4641,-1.0000) [circle,fill,inner sep=1.5pt]{};
\node at (5.1962,1.0000) [circle,fill,inner sep=1.5pt]{};
\node at (5.1962,2.0000) [circle,fill,inner sep=1.5pt]{};
\node at (4.3301,2.5000) [circle,fill,inner sep=1.5pt]{};
\node at (3.4641,2.0000) [circle,fill,inner sep=1.5pt]{};
\node at (3.4641,1.0000) [circle,fill,inner sep=1.5pt]{};
\node at (4.3301,0.5000) [circle,fill,inner sep=1.5pt]{};
\node at (6.0622,2.5000) [circle,fill,inner sep=1.5pt]{};
\node at (6.0622,3.5000) [circle,fill,inner sep=1.5pt]{};
\node at (5.1962,4.0000) [circle,fill,inner sep=1.5pt]{};
\node at (4.3301,3.5000) [circle,fill,inner sep=1.5pt]{};
\node at (4.3301,2.5000) [circle,fill,inner sep=1.5pt]{};
\node at (5.1962,2.0000) [circle,fill,inner sep=1.5pt]{};
\draw(-0.8660, -0.5000)--(-0.0000, -1.0000) node[pos=.5]{$d_0$};
\draw(0.0000, 1.0000)--(0.8660, 0.5000) node[pos=.5]{$d_1$};
\draw(0.8660, 2.5000)--(1.7321, 2.0000) node[pos=.5]{$d_2$};
\draw(1.7321, 4.0000)--(2.5981, 3.5000) node[pos=.5]{$d_3$};
\draw(0.8660, -0.5000)--(1.7321, -1.0000) node[pos=.5]{$d_0$};
\draw(1.7321, 1.0000)--(2.5981, 0.5000) node[pos=.5]{$d_1$};
\draw(2.5981, 2.5000)--(3.4641, 2.0000) node[pos=.5]{$d_2$};
\draw(3.4641, 4.0000)--(4.3301, 3.5000) node[pos=.5]{$d_3$};
\draw(2.5981, -0.5000)--(3.4641, -1.0000) node[pos=.5]{$d_0$};
\draw(3.4641, 1.0000)--(4.3301, 0.5000) node[pos=.5]{$d_1$};
\draw(4.3301, 2.5000)--(5.1962, 2.0000) node[pos=.5]{$d_2$};
\draw(5.1962, 4.0000)--(6.0622, 3.5000) node[pos=.5]{$d_3$};
\draw(-0.8660, -0.5000)--(-0.8660, 0.5000) node[pos=.5]{$c_0$};
\draw(0.0000, 1.0000)--(-0.0000, 2.0000) node[pos=.5]{$c_1$};
\draw(0.8660, 2.5000)--(0.8660, 3.5000) node[pos=.5]{$c_2$};
\draw(0.8660, -0.5000)--(0.8660, 0.5000) node[pos=.5]{$c_0$};
\draw(1.7321, 1.0000)--(1.7321, 2.0000) node[pos=.5]{$c_1$};
\draw(2.5981, 2.5000)--(2.5981, 3.5000) node[pos=.5]{$c_2$};
\draw(2.5981, -0.5000)--(2.5981, 0.5000) node[pos=.5]{$c_0$};
\draw(3.4641, 1.0000)--(3.4641, 2.0000) node[pos=.5]{$c_1$};
\draw(4.3301, 2.5000)--(4.3301, 3.5000) node[pos=.5]{$c_2$};
\draw(4.3301, -0.5000)--(4.3301, 0.5000) node[pos=.5]{$c_0$};
\draw(5.1962, 1.0000)--(5.1962, 2.0000) node[pos=.5]{$c_1$};
\draw(6.0622, 2.5000)--(6.0622, 3.5000) node[pos=.5]{$c_2$};
\draw(-0.8660, 0.5000)--(-0.0000, 1.0000) node[pos=.5]{$d_1$};
\draw(-0.0000, 2.0000)--(0.8660, 2.5000) node[pos=.5]{$d_2$};
\draw(0.8660, 3.5000)--(1.7321, 4.0000) node[pos=.5]{$d_3$};
\draw(-0.0000, -1.0000)--(0.8660, -0.5000) node[pos=.5]{$d_0$};
\draw(0.8660, 0.5000)--(1.7321, 1.0000) node[pos=.5]{$d_1$};
\draw(1.7321, 2.0000)--(2.5981, 2.5000) node[pos=.5]{$d_2$};
\draw(2.5981, 3.5000)--(3.4641, 4.0000) node[pos=.5]{$d_3$};
\draw(1.7321, -1.0000)--(2.5981, -0.5000) node[pos=.5]{$d_0$};
\draw(2.5981, 0.5000)--(3.4641, 1.0000) node[pos=.5]{$d_1$};
\draw(3.4641, 2.0000)--(4.3301, 2.5000) node[pos=.5]{$d_2$};
\draw(4.3301, 3.5000)--(5.1962, 4.0000) node[pos=.5]{$d_3$};
\draw(3.4641, -1.0000)--(4.3301, -0.5000) node[pos=.5]{$d_0$};
\draw(4.3301, 0.5000)--(5.1962, 1.0000) node[pos=.5]{$d_1$};
\draw(5.1962, 2.0000)--(6.0622, 2.5000) node[pos=.5]{$d_2$};
\end{tikzpicture}
    \end{minipage}
    $\qquad\times$
    \begin{minipage}{.4\textwidth}
    \begin{tikzpicture}
\node at (0.8660,-0.5000) [circle,fill,inner sep=1.5pt]{};
\node at (0.8660,0.5000) [circle,fill,inner sep=1.5pt]{};
\node at (0.0000,1.0000) [circle,fill,inner sep=1.5pt]{};
\node at (-0.8660,0.5000) [circle,fill,inner sep=1.5pt]{};
\node at (-0.8660,-0.5000) [circle,fill,inner sep=1.5pt]{};
\node at (-0.0000,-1.0000) [circle,fill,inner sep=1.5pt]{};
\node at (1.7321,1.0000) [circle,fill,inner sep=1.5pt]{};
\node at (1.7321,2.0000) [circle,fill,inner sep=1.5pt]{};
\node at (0.8660,2.5000) [circle,fill,inner sep=1.5pt]{};
\node at (-0.0000,2.0000) [circle,fill,inner sep=1.5pt]{};
\node at (0.0000,1.0000) [circle,fill,inner sep=1.5pt]{};
\node at (0.8660,0.5000) [circle,fill,inner sep=1.5pt]{};
\node at (2.5981,2.5000) [circle,fill,inner sep=1.5pt]{};
\node at (2.5981,3.5000) [circle,fill,inner sep=1.5pt]{};
\node at (1.7321,4.0000) [circle,fill,inner sep=1.5pt]{};
\node at (0.8660,3.5000) [circle,fill,inner sep=1.5pt]{};
\node at (0.8660,2.5000) [circle,fill,inner sep=1.5pt]{};
\node at (1.7321,2.0000) [circle,fill,inner sep=1.5pt]{};
\node at (2.5981,-0.5000) [circle,fill,inner sep=1.5pt]{};
\node at (2.5981,0.5000) [circle,fill,inner sep=1.5pt]{};
\node at (1.7321,1.0000) [circle,fill,inner sep=1.5pt]{};
\node at (0.8660,0.5000) [circle,fill,inner sep=1.5pt]{};
\node at (0.8660,-0.5000) [circle,fill,inner sep=1.5pt]{};
\node at (1.7321,-1.0000) [circle,fill,inner sep=1.5pt]{};
\node at (3.4641,1.0000) [circle,fill,inner sep=1.5pt]{};
\node at (3.4641,2.0000) [circle,fill,inner sep=1.5pt]{};
\node at (2.5981,2.5000) [circle,fill,inner sep=1.5pt]{};
\node at (1.7321,2.0000) [circle,fill,inner sep=1.5pt]{};
\node at (1.7321,1.0000) [circle,fill,inner sep=1.5pt]{};
\node at (2.5981,0.5000) [circle,fill,inner sep=1.5pt]{};
\node at (4.3301,2.5000) [circle,fill,inner sep=1.5pt]{};
\node at (4.3301,3.5000) [circle,fill,inner sep=1.5pt]{};
\node at (3.4641,4.0000) [circle,fill,inner sep=1.5pt]{};
\node at (2.5981,3.5000) [circle,fill,inner sep=1.5pt]{};
\node at (2.5981,2.5000) [circle,fill,inner sep=1.5pt]{};
\node at (3.4641,2.0000) [circle,fill,inner sep=1.5pt]{};
\node at (4.3301,-0.5000) [circle,fill,inner sep=1.5pt]{};
\node at (4.3301,0.5000) [circle,fill,inner sep=1.5pt]{};
\node at (3.4641,1.0000) [circle,fill,inner sep=1.5pt]{};
\node at (2.5981,0.5000) [circle,fill,inner sep=1.5pt]{};
\node at (2.5981,-0.5000) [circle,fill,inner sep=1.5pt]{};
\node at (3.4641,-1.0000) [circle,fill,inner sep=1.5pt]{};
\node at (5.1962,1.0000) [circle,fill,inner sep=1.5pt]{};
\node at (5.1962,2.0000) [circle,fill,inner sep=1.5pt]{};
\node at (4.3301,2.5000) [circle,fill,inner sep=1.5pt]{};
\node at (3.4641,2.0000) [circle,fill,inner sep=1.5pt]{};
\node at (3.4641,1.0000) [circle,fill,inner sep=1.5pt]{};
\node at (4.3301,0.5000) [circle,fill,inner sep=1.5pt]{};
\node at (6.0622,2.5000) [circle,fill,inner sep=1.5pt]{};
\node at (6.0622,3.5000) [circle,fill,inner sep=1.5pt]{};
\node at (5.1962,4.0000) [circle,fill,inner sep=1.5pt]{};
\node at (4.3301,3.5000) [circle,fill,inner sep=1.5pt]{};
\node at (4.3301,2.5000) [circle,fill,inner sep=1.5pt]{};
\node at (5.1962,2.0000) [circle,fill,inner sep=1.5pt]{};
\draw(-0.8660, -0.5000)--(-0.0000, -1.0000) node[pos=.5]{$c_0$};
\draw(0.0000, 1.0000)--(0.8660, 0.5000) node[pos=.5]{$c_0$};
\draw(0.8660, 2.5000)--(1.7321, 2.0000) node[pos=.5]{$c_0$};
\draw(1.7321, 4.0000)--(2.5981, 3.5000) node[pos=.5]{$c_0$};
\draw(0.8660, -0.5000)--(1.7321, -1.0000) node[pos=.5]{$c_1$};
\draw(1.7321, 1.0000)--(2.5981, 0.5000) node[pos=.5]{$c_1$};
\draw(2.5981, 2.5000)--(3.4641, 2.0000) node[pos=.5]{$c_1$};
\draw(3.4641, 4.0000)--(4.3301, 3.5000) node[pos=.5]{$c_1$};
\draw(2.5981, -0.5000)--(3.4641, -1.0000) node[pos=.5]{$c_2$};
\draw(3.4641, 1.0000)--(4.3301, 0.5000) node[pos=.5]{$c_2$};
\draw(4.3301, 2.5000)--(5.1962, 2.0000) node[pos=.5]{$c_2$};
\draw(5.1962, 4.0000)--(6.0622, 3.5000) node[pos=.5]{$c_2$};
\draw(-0.8660, -0.5000)--(-0.8660, 0.5000) node[pos=.5]{$d_0$};
\draw(0.0000, 1.0000)--(-0.0000, 2.0000) node[pos=.5]{$d_0$};
\draw(0.8660, 2.5000)--(0.8660, 3.5000) node[pos=.5]{$d_0$};
\draw(0.8660, -0.5000)--(0.8660, 0.5000) node[pos=.5]{$d_1$};
\draw(1.7321, 1.0000)--(1.7321, 2.0000) node[pos=.5]{$d_1$};
\draw(2.5981, 2.5000)--(2.5981, 3.5000) node[pos=.5]{$d_1$};
\draw(2.5981, -0.5000)--(2.5981, 0.5000) node[pos=.5]{$d_2$};
\draw(3.4641, 1.0000)--(3.4641, 2.0000) node[pos=.5]{$d_2$};
\draw(4.3301, 2.5000)--(4.3301, 3.5000) node[pos=.5]{$d_2$};
\draw(4.3301, -0.5000)--(4.3301, 0.5000) node[pos=.5]{$d_3$};
\draw(5.1962, 1.0000)--(5.1962, 2.0000) node[pos=.5]{$d_3$};
\draw(6.0622, 2.5000)--(6.0622, 3.5000) node[pos=.5]{$d_3$};
\draw(-0.8660, 0.5000)--(-0.0000, 1.0000) node[pos=.5]{$d_0$};
\draw(-0.0000, 2.0000)--(0.8660, 2.5000) node[pos=.5]{$d_0$};
\draw(0.8660, 3.5000)--(1.7321, 4.0000) node[pos=.5]{$d_0$};
\draw(-0.0000, -1.0000)--(0.8660, -0.5000) node[pos=.5]{$d_1$};
\draw(0.8660, 0.5000)--(1.7321, 1.0000) node[pos=.5]{$d_1$};
\draw(1.7321, 2.0000)--(2.5981, 2.5000) node[pos=.5]{$d_1$};
\draw(2.5981, 3.5000)--(3.4641, 4.0000) node[pos=.5]{$d_1$};
\draw(1.7321, -1.0000)--(2.5981, -0.5000) node[pos=.5]{$d_2$};
\draw(2.5981, 0.5000)--(3.4641, 1.0000) node[pos=.5]{$d_2$};
\draw(3.4641, 2.0000)--(4.3301, 2.5000) node[pos=.5]{$d_2$};
\draw(4.3301, 3.5000)--(5.1962, 4.0000) node[pos=.5]{$d_2$};
\draw(3.4641, -1.0000)--(4.3301, -0.5000) node[pos=.5]{$d_3$};
\draw(4.3301, 0.5000)--(5.1962, 1.0000) node[pos=.5]{$d_3$};
\draw(5.1962, 2.0000)--(6.0622, 2.5000) node[pos=.5]{$d_3$};
\end{tikzpicture}
    \end{minipage}
    }
    \caption{The weighted subsystem Hamiltonian $W_{\mathcal{B}_3}$ displayed in Figure \ref{fig:WB3honey} arises as the pointwise product of these two simpler weighting schemes; compare Figure \ref{fig:WB3gridproduct}.}
\end{center}
\end{figure}
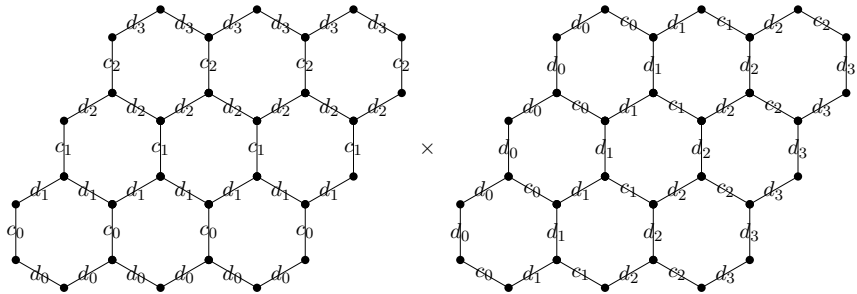
\vspace{-.5cm}

In the next few subsections, we use the weighted subsystem Hamiltonians to prove an analog of Proposition \ref{prop:keygrid} for the honeycomb lattice. Afterwards, in Subsections \ref{ssect:honeyoptimize}, we discuss how to choose the parameters $c_i$ and $d_i$ to optimize the relevant local gap threshold which will complete the proof of Theorem \ref{thm:mainhoney}.


\subsection{Expansion of the auxiliary operator}

We let $\mathcal T_L$ denote the set of possible translations of the plaquette $\mathcal B_\ell$ across $\mathbb H_L$, which we recall is the honeycomb lattice wrapped on an $L\times L$ torus. Given $t\in \mathcal T_L$, we write $W_{B_t}$ for the correspondingly translated weighted Hamiltonian. 

Combinatorial considerations parallel to those in Subsection \ref{ssect:expansionaux} yield, after implementing the symmetry requirement (ii) in Subsection \ref{ssect:coefficients},
\beq\label{eq:honeyauxnaive}
\begin{aligned}
\sum_{t\in\mathcal T_L} W_{B_t}^2 
    =&\left(\left(\sum_{i=0}^\ell d_i^2 \right)^2 - 2d_0^4 \right)\sum_{\begin{tikzpicture}
    \draw[thick](-.3,.15) -- (0,.3);
    \node at (-.3,.15)[circle,fill,inner sep=1pt]{};
    \node at (0,.3)[circle,fill,inner sep=1pt]{};
    \end{tikzpicture}} h_e \\
    &+ \left(\sum_{i=0}^\ell d_i^2 \right) \left(\sum_{i=0}^{\ell-1} c_i^2 \right)  \left(\sum_{\begin{tikzpicture}
    \draw[thick](-.3,.3) -- (0,.15);
    \node at (-.3,.3)[circle,fill,inner sep=1pt]{};
    \node at (0,.15)[circle,fill,inner sep=1pt]{};
    \end{tikzpicture}} h_e+\sum_{\begin{tikzpicture}
    \draw[thick](-.15,0) -- (-.15,.-.3);
    \node at (-.15,0)[circle,fill,inner sep=1pt]{};
    \node at (-.15,-.3)[circle,fill,inner sep=1pt]{};
    \end{tikzpicture}} h_e \right)\\
    &+ \left(\sum_{i=0}^{\ell-1} d_i^2\right) \left( \sum_{i=0}^{\ell-1} c_id_{i}\right) \left(\sum_{
     \begin{tikzpicture}
    \draw[thick](.1,0) -- (-.15,.-.15);
    \draw[thick](-.15,-.15)--(-.15,-.45);
    \node at (.1,0)[circle,fill,inner sep=1pt]{};
    \node at (-.15,-.15)[circle,fill,inner sep=1pt]{};
    \node at (-.15, -.45)[circle,fill,inner sep=1pt]{};
    \end{tikzpicture}
    }\{h_e, h_{e'}\}  +\sum_{
    \begin{tikzpicture}
    \draw[thick](.1,0) -- (.1,.-.3);
    \draw[thick](.1,-.3)--(-.15,-.45);
    \node at (.1,0)[circle,fill,inner sep=1pt]{};
    \node at (.1,-.3)[circle,fill,inner sep=1pt]{};
    \node at (-.15, -.45)[circle,fill,inner sep=1pt]{};
    \end{tikzpicture}
    }\{h_e, h_{e'}\} \right) 
    \\
     &+  \left( \sum_{i=0}^{\ell-1} c_id_i\right)^2 \left(\sum_{
    \begin{tikzpicture}
    \draw[thick](-.15,0) -- (-.15,.-.3);
    \draw[thick](-.15,-.3)--(.1,-.45);
    \node at (-.15,0)[circle,fill,inner sep=1pt]{};
    \node at (-.15,-.3)[circle,fill,inner sep=1pt]{};
    \node at (.1, -.45)[circle,fill,inner sep=1pt]{};
    \end{tikzpicture}
    }\{h_e, h_{e'}\}  +\sum_{
    \begin{tikzpicture}
    \draw[thick](-.15,0) -- (.1,.-.15);
    \draw[thick](.1,-.15)--(.1,-.45);
    \node at (-.15,0)[circle,fill,inner sep=1pt]{};
    \node at (.1,-.15)[circle,fill,inner sep=1pt]{};
    \node at (.1, -.45)[circle,fill,inner sep=1pt]{};
    \end{tikzpicture}
    }\{h_e, h_{e'}\} \right) 
    \\
     &+ \left(\left(\sum_{i=0}^\ell d_i^2\right) \left( \sum_{i=0}^{\ell-1} c_id_{i}\right)-c_0d_0^3\right) \left(\sum_{
    \begin{tikzpicture}
    \draw[thick](-.45,-.15) -- (-.2,0);
    \draw[thick](-.2,0)--(.05,-.15);
    \node at (-.45,-.15)[circle,fill,inner sep=1pt]{};
    \node at (-.2,0)[circle,fill,inner sep=1pt]{};
    \node at (.05, -.15)[circle,fill,inner sep=1pt]{};
    \end{tikzpicture}
    }\{h_e, h_{e'}\}  +\sum_{
    \begin{tikzpicture}
        \draw[thick](-.45,0) -- (-.2,-.15);
    \draw[thick](-.2,-.15)--(.05,0);
    \node at (-.45,0)[circle,fill,inner sep=1pt]{};
    \node at (-.2,-.15)[circle,fill,inner sep=1pt]{};
    \node at (.05, 0)[circle,fill,inner sep=1pt]{};
    \end{tikzpicture}
    }\{h_e, h_{e'}\} \right) \\
    &+R,
\end{aligned}
\eeq
Here, the various sums are taken over all edges with the indicated shape and $R$ denotes the remaining sums over anticommutators corresponding to nonadjacent edges. 

Notice that edges pointing in different directions do not appear with equal weights in the above expression. To remedy this, we consider the  following direction-averaged auxiliary operator. We denote $W_{B_l}^{(0)} = W_{B_l}$ and, rotating it by 60 degrees clockwise, respectively counterclockwise, we obtain $W_{B_l}^{(1)}$ and $W_{B_l}^{(2)}$. 
%

Summing formula \eqref{eq:honeyauxnaive} rotations and denoting $H_{\mathbb H_L}\equiv H_L$, we obtain
\beq\label{eq:honeyrecall}
K_1H_L^2 - \sum_{r=0}^2 \sum_{t\in\mathcal T_L} (W_{B_t}^{(r)})^2 \geq (K_1-K_0) H_L
\eeq
with the effective coefficients
\begin{align*}
    K_0&= 2\left(\sum_{i=0}^\ell d_i^2\right)\left(\sum_{i=0}^{\ell-1} c_i^2\right) + \left(\sum_{i=0}^\ell d_i^2\right)^2-2d_0^4 \\
    K_1&=\left(\sum_{i=0}^\ell d_i^2\right) \left( \sum_{i=0}^{\ell-1} c_id_{i}\right)+\left(\sum_{i=0}^{\ell-1} d_i^2\right) \left( \sum_{i=0}^{\ell-1} c_id_{i}\right)+ \left( \sum_{i=0}^{\ell-1}c_id_i\right)^2 -c_0d_0^3
\end{align*}
assuming that $K_1$ is larger than all coefficients in the remainder $R$. Under our choice of coefficients specified below, this will indeed be the case.


\subsection{Translation-invariant excited states}
We have the following analog of Lemma \ref{lm:translationgrid}.

\be{lm}\label{lm:translationhoney}
There exists a normalized state $|\phi\rangle$ satisfying $H_{\mathbb H_L}|\phi\rangle=\gam_L^{\mathrm{per}}|\phi\rangle$ such that
\begin{align*}
\langle \phi , (W_{B_t}^{(r)})^2  \phi \rangle &
\geq \gamma_\ell K_2 \langle \phi , W_{B_t}^{(r)} \phi \rangle,\qquad t\in\mathcal T_l,\quad r=0,1,2,
\end{align*}
with 
$$
K_2=\min\left\{\frac{1}{\ell^2+\ell}\left(\sum_{i=0}^\ell d_i\right)\left(\sum_{i=0}^{\ell-1} c_i\right) , \frac{1}{\ell^2+2\ell-1}\left(\left(\sum_{i=0}^\ell d_i\right)^2-2d_0^2\right)\right\}.
$$
\e{lm}

\be{proof}
This is a straightforward adaptation of the proof of Lemma \ref{lm:translationgrid}.
\e{proof}

We take the expectation of the operator inequality \eqref{eq:honeyrecall} in the state $|\phi\rangle$ and apply Lemma \ref{lm:translationhoney}. This gives
$$
\begin{aligned}
K_1 (\gamma^{\mathrm{per}}_L )^2
\geq& K_2 \gamma_\ell   \left\langle \phi, \sum_{r=0}^2 \sum_{t\in\mathcal T_L} W_{B_t}^{(r)} \phi \right\rangle
+ (K_1-K_0)\gamma^{\mathrm{per}}_L\\
=&K_3 \gamma_\ell \gamma^{\mathrm{per}}_L
+ (K_1-K_0)\gamma^{\mathrm{per}}_L
\end{aligned}
$$
where the second step follows from combinatorial considerations yielding the effective coefficient
$$
K_3 = K_2 \left(2\left(\sum_{i=0}^\ell d_i\right)\left(\sum_{i=0}^{\ell-1} c_i\right) + \left(\sum_{i=0}^\ell d_i\right)^2-2d_0^2 \right).
$$
This proves the following general finite-size criterion on the honeycomb lattice
\beq\label{eq:generalhoney}
\gamma_L^{\mathrm{per}} \geq \frac{K_3}{K_1}\left(\gamma_\ell - \frac{K_0-K_1}{K_3}\right).
\eeq

\subsection{Coefficient choice and conclusion}\label{ssect:honeyoptimize}
\be{proof}[Proof of Theorem \ref{thm:mainhoney}]
In view of \eqref{eq:generalhoney}, it remains to choose the coefficients $c_i,d_i$ so as to optimize the threshold $t_\ell=\frac{K_0-K_1}{K_3}$. Motivated by the considerations in the Euclidean situation, we set
\beq
\begin{aligned}
c_j &= \ell+(\ell-1)j-j^2 \\
d_j &= (1-\lambda)(\ell+1+\ell j-j^2) +  \lambda\left(\frac{(\ell+2)^2}{4}\right)
\end{aligned}
\eeq
but this time we choose the tilting parameter $\lambda = -\frac{30}{11\ell}$. (This is obtained by setting $ \lambda = \frac{C}{\ell}$ and optimizing the leading term of the resulting polynomial.) Note that the fact that $\lambda<0$ is no obstruction since for large enough $\ell$ we still have positivity $d_j>0$ as required . 

\textit{Asymptotics.}
Under this choice of $\lambda$, we calculate the leading terms of the relevant ratios of polynomials and find
\beq\label{eq:K0K1honey}
\begin{aligned}
    K_0 - K_1 &= \frac{19}{1980}\ell^{8} + \frac{1076}{5445}\ell^7 + O(\ell^6) \\
    K_3 &= \frac{1}{432}\ell^{10} + \frac{13}{264}\ell^{9} + O(\ell^8).
\end{aligned}
\eeq
which yields a local gap threshold of the form
$$
\frac{K_0-K_1}{K_3}\geq \frac{228}{55\ell^2} + O\left(\frac{1}{\ell^3}\right),
$$
which has the claimed leading asymptotics.

\textit{Error estimates.}
We can obtain an explicit, still rather crude error estimate by inspecting the coefficients of the polynomials. This yields for the local gap threshold
$$
\frac{K_0-K_1}{K_3}\leq  \frac{228}{55\ell^2} + \frac{108}{\ell^3}
$$
Finally, we obtain \eqref{eqn:hexagongap} in Theorem \ref{thm:mainhoney} from the bound $\frac{K_3}{K_1} \geq \frac{1}{2}$. The values displayed in Table \ref{tab:table3} are obtained by explicit optimization of $\lambda$ for small $\ell$-values subject to the given constraints.
\e{proof}

\section{Proof of Theorem \ref{thm:maintri} on the triangular lattice}
We follow the line of argumentation for Theorems \ref{thm:main} and \ref{thm:mainhoney}. We focus on the key differences and omit other details.


To define the weighted subsystem Hamiltonians we use the weighting scheme depicted in Figure \ref{fig:WB3tri}, which again can be viewed as arising from taking the products of edges from simpler Hamiltonians. 

\vspace{.5cm}
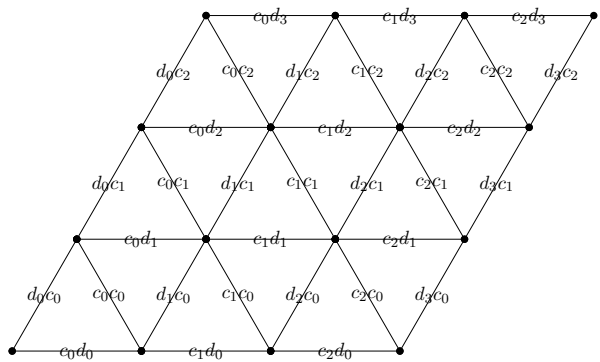
\begin{figure}[h]
\begin{center}
\resizebox{8cm}{!}{
\begin{tikzpicture}
\node at (2.5981,-0.7500) [circle,fill,inner sep=1.5pt]{};
\node at (1.2990,1.5000) [circle,fill,inner sep=1.5pt]{};
\node at (0.0000,-0.7500) [circle,fill,inner sep=1.5pt]{};
\node at (2.5981,-0.7500) [circle,fill,inner sep=1.5pt]{};
\node at (1.2990,1.5000) [circle,fill,inner sep=1.5pt]{};
\node at (-0.0000,-0.7500) [circle,fill,inner sep=1.5pt]{};
\node at (3.8971,1.5000) [circle,fill,inner sep=1.5pt]{};
\node at (2.5981,3.7500) [circle,fill,inner sep=1.5pt]{};
\node at (1.2990,1.5000) [circle,fill,inner sep=1.5pt]{};
\node at (3.8971,1.5000) [circle,fill,inner sep=1.5pt]{};
\node at (2.5981,3.7500) [circle,fill,inner sep=1.5pt]{};
\node at (1.2990,1.5000) [circle,fill,inner sep=1.5pt]{};
\node at (5.1962,3.7500) [circle,fill,inner sep=1.5pt]{};
\node at (3.8971,6.0000) [circle,fill,inner sep=1.5pt]{};
\node at (2.5981,3.7500) [circle,fill,inner sep=1.5pt]{};
\node at (5.1962,3.7500) [circle,fill,inner sep=1.5pt]{};
\node at (3.8971,6.0000) [circle,fill,inner sep=1.5pt]{};
\node at (2.5981,3.7500) [circle,fill,inner sep=1.5pt]{};
\node at (5.1962,-0.7500) [circle,fill,inner sep=1.5pt]{};
\node at (3.8971,1.5000) [circle,fill,inner sep=1.5pt]{};
\node at (2.5981,-0.7500) [circle,fill,inner sep=1.5pt]{};
\node at (5.1962,-0.7500) [circle,fill,inner sep=1.5pt]{};
\node at (3.8971,1.5000) [circle,fill,inner sep=1.5pt]{};
\node at (2.5981,-0.7500) [circle,fill,inner sep=1.5pt]{};
\node at (6.4952,1.5000) [circle,fill,inner sep=1.5pt]{};
\node at (5.1962,3.7500) [circle,fill,inner sep=1.5pt]{};
\node at (3.8971,1.5000) [circle,fill,inner sep=1.5pt]{};
\node at (6.4952,1.5000) [circle,fill,inner sep=1.5pt]{};
\node at (5.1962,3.7500) [circle,fill,inner sep=1.5pt]{};
\node at (3.8971,1.5000) [circle,fill,inner sep=1.5pt]{};
\node at (7.7942,3.7500) [circle,fill,inner sep=1.5pt]{};
\node at (6.4952,6.0000) [circle,fill,inner sep=1.5pt]{};
\node at (5.1962,3.7500) [circle,fill,inner sep=1.5pt]{};
\node at (7.7942,3.7500) [circle,fill,inner sep=1.5pt]{};
\node at (6.4952,6.0000) [circle,fill,inner sep=1.5pt]{};
\node at (5.1962,3.7500) [circle,fill,inner sep=1.5pt]{};
\node at (7.7942,-0.7500) [circle,fill,inner sep=1.5pt]{};
\node at (6.4952,1.5000) [circle,fill,inner sep=1.5pt]{};
\node at (5.1962,-0.7500) [circle,fill,inner sep=1.5pt]{};
\node at (7.7942,-0.7500) [circle,fill,inner sep=1.5pt]{};
\node at (6.4952,1.5000) [circle,fill,inner sep=1.5pt]{};
\node at (5.1962,-0.7500) [circle,fill,inner sep=1.5pt]{};
\node at (9.0933,1.5000) [circle,fill,inner sep=1.5pt]{};
\node at (7.7942,3.7500) [circle,fill,inner sep=1.5pt]{};
\node at (6.4952,1.5000) [circle,fill,inner sep=1.5pt]{};
\node at (9.0933,1.5000) [circle,fill,inner sep=1.5pt]{};
\node at (7.7942,3.7500) [circle,fill,inner sep=1.5pt]{};
\node at (6.4952,1.5000) [circle,fill,inner sep=1.5pt]{};
\node at (10.3923,3.7500) [circle,fill,inner sep=1.5pt]{};
\node at (9.0933,6.0000) [circle,fill,inner sep=1.5pt]{};
\node at (7.7942,3.7500) [circle,fill,inner sep=1.5pt]{};
\node at (10.3923,3.7500) [circle,fill,inner sep=1.5pt]{};
\node at (9.0933,6.0000) [circle,fill,inner sep=1.5pt]{};
\node at (7.7942,3.7500) [circle,fill,inner sep=1.5pt]{};
\node at (11.6913,6.0000) [circle,fill,inner sep=1.5pt]{};
\draw(0.0000, -0.7500)--(2.5981, -0.7500) node[pos=.5]{$c_0d_0$};
\draw(1.2990, 1.5000)--(3.8971, 1.5000) node[pos=.5]{$c_0d_1$};
\draw(2.5981, 3.7500)--(5.1962, 3.7500) node[pos=.5]{$c_0d_2$};
\draw(3.8971, 6.0000)--(6.4952, 6.0000) node[pos=.5]{$c_0d_3$};
\draw(2.5981, -0.7500)--(5.1962, -0.7500) node[pos=.5]{$c_1d_0$};
\draw(3.8971, 1.5000)--(6.4952, 1.5000) node[pos=.5]{$c_1d_1$};
\draw(5.1962, 3.7500)--(7.7942, 3.7500) node[pos=.5]{$c_1d_2$};
\draw(6.4952, 6.0000)--(9.0933, 6.0000) node[pos=.5]{$c_1d_3$};
\draw(5.1962, -0.7500)--(7.7942, -0.7500) node[pos=.5]{$c_2d_0$};
\draw(6.4952, 1.5000)--(9.0933, 1.5000) node[pos=.5]{$c_2d_1$};
\draw(7.7942, 3.7500)--(10.3923, 3.7500) node[pos=.5]{$c_2d_2$};
\draw(9.0933, 6.0000)--(11.6913, 6.0000) node[pos=.5]{$c_2d_3$};
\draw(1.2990, 1.5000)--(0.0000, -0.7500) node[pos=.5]{$d_0c_0$};
\draw(2.5981, 3.7500)--(1.2990, 1.5000) node[pos=.5]{$d_0c_1$};
\draw(3.8971, 6.0000)--(2.5981, 3.7500) node[pos=.5]{$d_0c_2$};
\draw(3.8971, 1.5000)--(2.5981, -0.7500) node[pos=.5]{$d_1c_0$};
\draw(5.1962, 3.7500)--(3.8971, 1.5000) node[pos=.5]{$d_1c_1$};
\draw(6.4952, 6.0000)--(5.1962, 3.7500) node[pos=.5]{$d_1c_2$};
\draw(6.4952, 1.5000)--(5.1962, -0.7500) node[pos=.5]{$d_2c_0$};
\draw(7.7942, 3.7500)--(6.4952, 1.5000) node[pos=.5]{$d_2c_1$};
\draw(9.0933, 6.0000)--(7.7942, 3.7500) node[pos=.5]{$d_2c_2$};
\draw(9.0933, 1.5000)--(7.7942, -0.7500) node[pos=.5]{$d_3c_0$};
\draw(10.3923, 3.7500)--(9.0933, 1.5000) node[pos=.5]{$d_3c_1$};
\draw(11.6913, 6.0000)--(10.3923, 3.7500) node[pos=.5]{$d_3c_2$};
\draw(1.2990, 1.5000)--(2.5981, -0.7500) node[pos=.5]{$c_0c_0$};
\draw(2.5981, 3.7500)--(3.8971, 1.5000) node[pos=.5]{$c_0c_1$};
\draw(3.8971, 6.0000)--(5.1962, 3.7500) node[pos=.5]{$c_0c_2$};
\draw(3.8971, 1.5000)--(5.1962, -0.7500) node[pos=.5]{$c_1c_0$};
\draw(5.1962, 3.7500)--(6.4952, 1.5000) node[pos=.5]{$c_1c_1$};
\draw(6.4952, 6.0000)--(7.7942, 3.7500) node[pos=.5]{$c_1c_2$};
\draw(6.4952, 1.5000)--(7.7942, -0.7500) node[pos=.5]{$c_2c_0$};
\draw(7.7942, 3.7500)--(9.0933, 1.5000) node[pos=.5]{$c_2c_1$};
\draw(9.0933, 6.0000)--(10.3923, 3.7500) node[pos=.5]{$c_2c_2$};
 \end{tikzpicture}

}
\caption{The weighted subsystem Hamiltonian $W_{\mathcal{B}_3}$ on the slanted subgrid of the triangular lattice. Edges are labeled by their weights. This is the triangular-lattice analog of Figure \ref{fig:WB3grid}.}
\label{fig:WB3tri}
\end{center}
\end{figure}
\vspace{-.5cm}

As in the honeycomb case, we define the appropriate auxiliary operator by summing over squares of translated and rotated copies of these weighted Hamiltonians, call them $\{W_{B_t}^{(0)}, W_{B_t}^{(1)}, W_{B_t}^{(2)}\}$, where $t$ labels the possible translation in $\mathcal T_L$. We denote $H_{\mathbb T_L}\equiv H_L$. Combinatorial considerations yield
\begin{align*}
    \sum_{t\in\mathcal T_L}\sum_{r=0,1,2} (W_{B_l}^{(r)})^2 &= K_0 H_L
+ K_1\sum_{\begin{tikzpicture}
\draw[thick] (0,.15) -- (.3,.15);
\draw[thick] (.3,.15) -- (.6,.15);
\tkzDefPoint(0,.15){A}
\tkzDefPoint(.3,.15){B}
\tkzDefPoint(.6,.15){C}
\foreach \n in {A,B,C}
  \node at (\n)[circle,fill,inner sep=1pt]{};
\end{tikzpicture}} \{h_e, h_{e'}\} + K_2\sum_{\begin{tikzpicture}
\draw[thick](-.15, -.3) -- (.15, -.3);
\draw[thick](-.15, -.3) -- (0, -.03);
\node at (-.15,-.3)[circle,fill,inner sep=1pt]{};
\node at (.15,-.3)[circle,fill,inner sep=1pt]{};
\node at (0,-.03)[circle,fill,inner sep=1pt]{};
\end{tikzpicture}} \{h_e, h_{e'}\} 
+ K_3 \sum_{\begin{tikzpicture}
\draw[thick](-.15, -.3) -- (.15, -.3);
\draw[thick](-.15, -.3) -- (-.3, -.03);
\node at (-.15,-.3)[circle,fill,inner sep=1pt]{};
\node at (.15,-.3)[circle,fill,inner sep=1pt]{};
\node at (-.3,-.03)[circle,fill,inner sep=1pt]{};
\end{tikzpicture}} \{h_e, h_{e'}\} + R.
\end{align*}
Here the sums are over all edge pairings of the associated shape (irrespective of orientation) and the term $R$ contains anti-commutators over non-adjacent edges, which are positive semidefinite. Moreover, we introduced the effective coefficients
\begin{align*}
    K_0 &= 2\left( \sum_{i=0}^\ell d_i^2 \right) \left( \sum_{i=0}^{\ell-1}c_i^2\right) + \left( \sum_{i=0}^{\ell-1}c_i^2\right)^2 \\
    K_1 &= 2\left( \sum_{i=0}^\ell d_i^2 \right)\left( \sum_{i=0}^{\ell-2} c_ic_{i+1}\right) + \left( \sum_{i=0}^{\ell-2} c_ic_{i+1}\right)^2 \\
    K_2 &= 2 \left( \sum_{i=0}^{\ell-1}c_i^2\right)\left( \sum_{i=0}^{\ell-1} c_id_i\right) + \left( \sum_{i=0}^{\ell-1} c_id_i\right)^2 \\
    K_3 &= 2 \left(\sum_{i=0}^{\ell-2} c_ic_{i+1} \right)\left( \sum_{i=0}^{\ell-1} c_id_i\right) + \left( \sum_{i=0}^{\ell-1} c_id_i\right)^2
\end{align*}

Assuming that $K_2 \geq K_1, K_3$ and using the operator Cauchy-Schwarz inequality, we obtain
\beq\label{eq:recalltri}
\begin{aligned}
K_2 H_L^2 - \sum_{r=0}^2 \sum_t (W_{B_t}^{(r)})^2 &\geq (K_2-K_0)H_L + (K_2-K_1) \sum_{\begin{tikzpicture}
\draw[thick] (0,.15) -- (.3,.15);
\draw[thick] (.3,.15) -- (.6,.15);
\tkzDefPoint(0,.15){A}
\tkzDefPoint(.3,.15){B}
\tkzDefPoint(.6,.15){C}
\foreach \n in {A,B,C}
  \node at (\n)[circle,fill,inner sep=1pt]{};
\end{tikzpicture}}  \{h_e, h_{e'}\} 
+ (K_2-K_3)\sum_{\begin{tikzpicture}
\draw[thick](-.15, -.3) -- (.15, -.3);
\draw[thick](-.15, -.3) -- (-.3, -.03);
\node at (-.15,-.3)[circle,fill,inner sep=1pt]{};
\node at (.15,-.3)[circle,fill,inner sep=1pt]{};
\node at (-.3,-.03)[circle,fill,inner sep=1pt]{};
\end{tikzpicture}} \{h_e, h_{e'}\} \\
&\geq  (K_2-K_0)H_L - 2(K_2-K_1)H_L -4(K_2-K_3)H_L\\
&= (2K_1+4K_3-K_0-5K_2) H_L,
\end{aligned}
\eeq
where in the first step we also discarded the terms in $R$ using the assumption that the $c_i,d_i$ are increasing up until their midpoints, and hence will have coefficients less than $K_1,K_2,K_3$. 

By adapting the proof of Lemma \ref{lm:translationgrid} to the triangular lattice, we obtain a normalized state $|\phi\rangle$ in the $\gamma_L^{\mathrm{per}}$-eigenspace of $H_L$ such that
\beq\label{eq:translationtri}
\langle \phi | (W_{B_t}^{(r)})^2 | \phi \rangle 
\geq  
K_4 \gamma_\ell \langle \phi | W_{B_l}^{(r)} | \phi \rangle,
\qquad t\in\mathcal T_L,\quad r=0,1,2,
\eeq
with the new effective coefficient
$$
K_4=\min\left\{\frac{1}{\ell^2+\ell}\left(\sum_{i=0}^\ell d_i\right)\left(\sum_{i=0}^{\ell-1} c_i\right) , \frac{1}{\ell^2}\left(\sum_{i=0}^{\ell-1} c_i\right)^2\right\}.
$$
Taking the expectation of \eqref{eq:recalltri} in the state $|\phi\rangle$, using \eqref{eq:translationtri}, and calculating $ \sum_{r=0}^2 \sum_t W_{B_t}^{(r)}$ yields
$$
K_2 \gamma_L^{\mathrm{per}}\geq K_5 \gamma_\ell +2K_1+4K_3-K_0-5K_2
$$
with 
$$
K_5 = K_4\left(2\left(\sum_{i=0}^\ell d_i\right)\left(\sum_{i=0}^{\ell-1} c_i\right) +  \left(\sum_{i=0}^{\ell-1} c_i\right)^2\right).
$$
This yields the general form of the finite-size criterion
$$
\gamma_L^{\mathrm{per}} \geq \frac{K_5}{K_2}\left(\gamma_\ell - \left(\frac{K_0+5K_2-2K_1-4K_3}{K_5}\right)\right)
$$

It remains to choose the coefficients subject to the constraints (i)-(iii) and $K_2 \geq K_1, K_3$. Similarly to before, we take
\beq\label{eq:ciditri}
\begin{aligned}
c_j &= \ell+(\ell-1)j-j^2 \\
d_j &= (1-\lambda)(\ell+1+\ell j-j^2)  + \lambda\left(\frac{(\ell+2)^2}{4}\right)
\end{aligned}
\eeq
but now with $\lambda =  \frac{-30 + 20 \sqrt{5}}{11 \ell}$, again a parameter chosen through optimizing the leading coefficient of the relevant polynomials.

With explicit error estimates for the subleading terms, we obtain
$$
\gamma_L^{\mathrm{per}} \geq \frac{1}{2}\left(\gamma_\ell - \left(\frac{144}{5\ell^2} + \frac{432}{\ell^3}\right) \right) 
$$
as claimed. To obtain the values in Table \ref{tab:table4}, we instead use \eqref{eq:ciditri} and optimize over $\lambda$ for fixed $\ell$ subject to the given constraints. This completes the proof of Theorem \ref{thm:maintri}.
\qed

\end{document}